%% file: main.tex
\newif\ifllncs\llncsfalse

\newif\iffull\fulltrue

\newif\ifblind\blindfalse

\newif\ifnotes\notesfalse

\ifllncs \documentclass[runningheads,a4paper,envcountsect,envcountsame]{llncs/llncs} 
\else \documentclass[11pt]{article} \fi 

\input{preamble}

\title{A Complexity-Theoretic Approach to Proofs of Space}
\ifblind
\author{}
\date{}
\else
\author{
        Marshall Ball\thanks{New York University. Email: \texttt{marshall.ball@cs.nyu.edu}. ORCID: \texttt{0000-0002-4236-3710}.}
    \and
        Jiaxin Guan\thanks{Zellic. Email: \texttt{jiaxin@guan.io}. ORCID: \texttt{0000-0003-1823-8845}.}
    }
\fi

\begin{document}

\maketitle
\begin{abstract}
    A Proof of Space, PoS, as introduced by Dziembowski \textit{et al.} [CRYPTO'15], is a two-phase protocol that enables a Prover to convince an efficient Verifier that it has allocated a large amount of persistent memory to storing some information.

    To our knowledge, all existing PoS protocols are only known to be secure in the random oracle model (or under ad hoc assumptions about cryptographic assumptions). We provide an elementary framework for constructing PoS from a combination of derandomization assumptions and cryptographic assumptions.

    We provide a few simple instantiations of the framework. We show that non-trivial PoS follow from (a) $\mathsf{E}=\mathsf{DTIME[2^{O(n)}]}$ is hard for exponential-size nondeterministic circuits (an assumption introduced to show $\mathsf{AM}=\mathsf{NP}$), and (b) collision-resistant hash functions. We also show that PoS with nearly optimal parameters and interaction pattern follows from assumption (a) above and (c) SNARGs for $\mathsf{P}$.
\end{abstract}

\noteswarning 
\newif\ifsubfilebib\subfilebibfalse  
\mnote{TODO:\begin{itemize}
    \item section 6 sound proof clarify [Jiaxin]
    \item Add $\varepsilon$ $\delta/\rho$ relationship for ECC corollary. For any $\varepsilon\in(0, 1/2)$, there exists $\rho$... [Jiaxin]
    \item proofread section 5 + adapt to covering lemma (remove old covering lemma) [Marshall]
    \item proofread intro [Jiaxin] + add informal theorem
\end{itemize}}

\subfile{intro}

\subfile{prelim}

\subfile{ktcomp}

\subfile{extraction}

\subfile{pos-2}

\bibliographystyle{alpha}
\bibliography{Bibliography/abbrev3,Bibliography/crypto,Bibliography/biblio}

\end{document}

%% file: preamble.tex
\usepackage{amsmath,amsthm,amsfonts,amssymb}
\usepackage{fullpage}
\usepackage{hyperref}
\usepackage{xcolor}
\usepackage[english]{babel}
\usepackage[utf8x]{inputenc}
\usepackage{tcolorbox}
\usepackage{subfiles}
\usepackage{mdframed}
\usepackage[normalem]{ulem}
\usepackage[shortlabels]{enumitem}

\theoremstyle{definition}
\newtheorem{definition}{Definition}

\theoremstyle{remark}
\newtheorem{remark}{Remark}

\theoremstyle{plain}
\newtheorem{theorem}{Theorem}
\newtheorem{lemma}{Lemma}
\newtheorem{prop}{Proposition}
\newtheorem{corollary}{Corollary}
\newtheorem{construction}{Construction}

\ifnotes 
\usepackage{color}
\newcommand{\noteby}[3]{{\textcolor{#3}{\footnotesize{\bf (#1:} {#2}{\bf )}}}}
\newcommand{\noteswarning}{{\begin{center} {\Large WARNING: NOTES ON}\end{center}}}
\else 
\newcommand{\noteby}[3]{{}}
\newcommand{\noteswarning}{{}}
\fi 
\newcommand{\mnote}[1]{{\noteby{Marshall}{#1}{red!70}}}

\def\ddefloop#1{\ifx\ddefloop#1\else\ddef{#1}\expandafter\ddefloop\fi}

\def\ddef#1{\expandafter\def\csname bb#1\endcsname{\ensuremath{\mathbb{#1}}}}
\ddefloop ABCDEFGHIJKLMNOPQRSTUVWXYZ\ddefloop

\def\ddef#1{\expandafter\def\csname c#1\endcsname{\ensuremath{\mathcal{#1}}}}
\ddefloop ABCDEFGHIJKLMNOPQRSTUVWXYZ\ddefloop

\def\ddef#1{\expandafter\def\csname v#1\endcsname{\ensuremath{\boldsymbol{#1}}}}
\ddefloop ABCDEFGHIJKLMNOPQRSTUVWXYZabcdefghijlmnopqrstuvwxyz\ddefloop

\def\ddef#1{\expandafter\def\csname v#1\endcsname{\ensuremath{\boldsymbol{\csname #1\endcsname}}}}
\ddefloop {alpha}{beta}{gamma}{delta}{epsilon}{varepsilon}{zeta}{eta}{theta}{vartheta}{iota}{kappa}{lambda}{mu}{nu}{xi}{pi}{varpi}{rho}{varrho}{sigma}{varsigma}{tau}{upsilon}{phi}{varphi}{chi}{psi}{omega}{Gamma}{Delta}{Theta}{Lambda}{Xi}{Pi}{Sigma}{varSigma}{Upsilon}{Phi}{Psi}{Omega}{ell}\ddefloop

\newcommand{\ignore}[1]{}

\newcommand{\poly}{\mathsf{poly}}
\newcommand{\polylog}{\mathsf{polylog}}
\newcommand{\negl}{\mathsf{negl}}

\newcommand{\from}{\leftarrow}

\newcommand{\zo}{\{0,1\}}

\newcommand{\out}{\mathsf{out}}
\newcommand{\inp}{\mathsf{in}}

\newcommand{\params}{\mathsf{params}}

\newcommand{\cf}{\mathsf{cf}}

\newcommand{\Gen}{\mathsf{Gen}}
\newcommand{\Prove}{\mathsf{Prove}}
\newcommand{\Verify}{\mathsf{Verify}}

\newcommand{\accept}{\mathsf{accept}}

\newcommand{\pk}{\mathsf{pk}}
\newcommand{\vk}{\mathsf{vk}}

\newcommand{\barT}{\overline{T}}

\newcommand{\ARG}{\mathsf{ARG}}
\newcommand{\ECC}{\mathsf{ECC}}
\newcommand{\En}{\mathsf{En}}
\newcommand{\De}{\mathsf{De}}

\newcommand{\Commit}{\mathsf{Commit}}
\newcommand{\Open}{\mathsf{Open}}

\newcommand{\PCP}{\mathsf{PCP}}

\newcommand{\crs}{\mathsf{crs}}
\newcommand{\pf}{\mathsf{pf}}
\newcommand{\ans}{\mathsf{ans}}

\newcommand{\lin}{{\ell_{\mathsf{in}}}}
\newcommand{\lout}{{\ell_{\mathsf{out}}}}

\newcommand{\MERKLE}{\mathsf{MT}}

\newcommand{\st}{\mathsf{st}}

\newcommand{\good}{\mathrm{Good}}

\newcommand{\nin}{{n_{\mathsf{in}}}}
\newcommand{\nout}{{n_{\mathsf{out}}}}

\newcommand{\Tau}{T}

\newcommand{\aK}{\mathsf{aK}}
\newcommand{\pK}{\mathsf{pK}}
\newcommand{\Kt}{\mathsf{K}^{t}}
\newcommand{\KT}{\mathsf{K}^{T}}
\newcommand{\pKt}{\pK^{t}}
\newcommand{\aKt}{\aK^{t}}

\newcommand{\great}{\mathrm{Great}}

%% file: intro.tex
\section{Introduction}
A Proof of Space protocol, introduced by Dziembowski {\em et al.}~\cite{C:DFKP15}, is a two-phase protocol that enables a very efficient Verifier to check that a Prover is using a lot of storage (persistently).\footnote{A different definition with the same name emerged at the same time, due to Ateniese et al.~\cite{SCN:ABFG14}. This latter notion is not the subject of this paper.} The first phase consists of a single short verifier message, but takes place over a relatively long period of time. After the first phase, the Prover is supposed to have stored some data $\sigma$ of size $N$. In the second phase, the Verifier interacts with the Prover to check that prover has indeed committed $\approx N$ bits of random access memory to this task. It is integral that this second phase takes place over a short interval of time, else the Prover could simply expand the succinct description of $\sigma$ at that moment to respond to the Verifier's queries.

A rich line of literature~\cite{C:DFKP15,SCN:ABFG14,TCC:RenDev16,EC:ACKKPT16,EC:AlwBloPie17,EC:AlwBloPie18,AC:AACKPR17,EC:Fisch19,FOCS:Reyzin24} has emerged following Dziembowski {\em et al.}'s seminal work. These works contain elegant theoretical constructions and connections to graph pebbling, however (to our knowledge) they all \emph{critically} rely on the random oracle model to work. These proofs critically make simultaneous use of both the incompressibility and extractibility of random oracle queries, and we do not know any constructions outside of the random oracle model. In fact, we do not even know of any clean (even non-falsifiable) assumptions about cryptographic hash functions that would suffice for the provable security of these constructions.\footnote{For the related but incomparable case of memory-hard functions, Ameri, Block and Blocki~\cite{SCN:AmeBloBlo22} gave a construction assuming memory-hard puzzles, iO, and OWF (following a template laid out by Bitansky {\em et al.}~\cite{ITCS:BGJPVW16} for time-lock puzzles).}

In this work, we present a very simple approach to constructing Proofs of Space in the plain model. We provide three straightforward instantiations of our framework with various trade-offs. Before continuing, we emphasize that our focus is on elementary, pedagogically clean, generic feasibility; not sophisticated threat models and practical protocols.

\subsection{This Work: An Elementary Approach to Proofs of Space}
Perhaps the key difficulty in constructing Proofs of Space (PoS) without a random oracle, is that one needs a long string $\sigma$ that is incompressible, and yet it should paradoxically also have a short description! The only reason this is not impossible is the time constraints on the PoS phases: the first phase is long (time $T$), while the second is short (time $t$). So, we actually have two different fine-grained notions of (in)compressibility:
\begin{enumerate}
    \item $\sigma$ can be expanded from a very short description in time $T$,
    \item $\sigma$ \emph{cannot} be expanded from any description slightly shorter than $|\sigma|$ in time $t$, even when given the short description.
\end{enumerate}
Thankfully, this in fact nearly corresponds to a notion in complexity theory, \emph{computational depth} introduced by Antunes {\em et al.}~\cite{TCS:AFvMV06}. In essence, any PoS necessarily implies an explicit means of finding strings with a (fine-grained and computationally-sound) notion of computational depth: strings that can be highly compressed (if decompression is allowed to be inefficient) but don't admit ``easy'' compressions (if decompression must be fast).\footnote{Technically, in PoS we simply require that it is hard to find a quickly (time $t$) decompressible description of $\sigma$ (Yao-style incompressibility), whereas computational depth says that no such description exists. Antunes et al. defined a few notions of computational depth, but in all there is no time bound $T$ in the short description.} To explain this notion, we need to first introduce some preliminaries. 

\paragraph{Background: time-bounded Kolmogorov complexity.} We require a classic notion of compressibility: the \emph{$t$-time-bounded Kolmogorov complexity} of a string $x$, $\Kt(x)$, is the length of the shortest string $s$ such that $U:s\stackrel{t}{\rightsquigarrow}x$, i.e.~the universal Turing machine $U$ on input $s$ produces $x$ in at most $t$ steps.\footnote{The choice of $U$ only matters up to an additive constant.} The \emph{conditional} $t$-time-bounded Kolmogorov complexity of $x$ given $y$, $\Kt(x|y)$, is the length of the shortest string $s$ such that $U:(s,y)\stackrel{t}{\rightsquigarrow}x$.

\paragraph{Functions of high computational depth.} We are interested not just in the existence of strings with high computational depth, but also an efficient means of finding them. With this in mind, we are interested in \emph{functions with high computational depth}. 

We say a function $f:\zo^k\to\zo^N$ has $(t,T)$ \emph{computational depth} $\alpha N$ with error $\varepsilon$ if
\begin{enumerate}
    \item $f$ is computable in time $T$,
    \item With probability at most $\varepsilon$ over a random $x\stackrel{u}{\gets} \zo^k$,
    \[ \Kt(f(x)|x)) \le \alpha N.\]
\end{enumerate}
In other words, $\KT(f(x))\le |x|=k$ and yet $\Kt(f(x)|x)\approx N$.\footnote{A straight-forward fine-grained extension of Antunes et al.'s definition of \emph{computational depth} of a string $x$ is the difference between the length of its $t$-time compression and its $T$-time compression:$\mathsf{CD}^{t,T}(x) = \Kt(x) - \KT(x)$.} In our setting we want $k$ to be as small as possible, $k=N^\gamma$ for some  constant $\gamma<1$ or even $k=O(\log n)$.

In fact, in our setting of small $k=|x|$, we can observe that it suffices that $\Kt(f(x)) >\alpha N$ as $\Kt(f(x)|x) \ge \Kt(f(x)) - |x|$.

We will actually require a more robust notion in two respects.

Firstly, that $f$ has high \emph{approximate} computational depth, namely that for most $x$ $\Kt(\hat{y}|x)>\alpha N$ for any $y$ close to $x$ in hamming distance. This is easily achieved by composing with an (explicit) list-decodable code, $(\mathsf{E},\mathsf{D})$: $f'(x) = \mathsf{E}(f(x))$. However, our construction methodology directly gives this with essentially optimal parameters.

Secondly, that $f(x)$ remains incompressible even in the presence of \emph{auxiliary information} (i.e.~in the non-uniform setting). This roughly says that for any non-uniform time $t$ (decompression) algorithm $\cA$, except with negligible probability (over $x$) there is no short $\sigma$ such that $\cA(\sigma)=f(x)$.
This robust notion is tailored to a cryptographic setting where we wish for non-uniform security. For example, an adversary may perform some expensive preprocessing to use at execution time. We note that our constructions of $f$ immediately achieve this stronger non-uniform notion.

Having defined a notion that intuitively seems to capture a necessary element of an explicit PoS, the natural next question is how do we construct this function of high computational depth?

We observe that such a function (with high stretch) can be easily lifted from classic derandomization techniques.

\paragraph{Background: PRGs for nondeterministic distinguishers.} The key ingredient in our construction is the notion of a \emph{PRG against nondeterministic distinguishers}~\cite{DBLP:journals/cc/KinneMS12,CCC:CU05}, $G:\zo^k\to\zo^n$. We say that $G$ is such a seed-extending PRG against nondeterministic $t(n)=\poly(n)$ with error $\varepsilon$, if it is computable in time $T=t^{O(1)}=\poly(n)$ time, and fools any nondeterministic $t(n)$-time distinguisher, $D$:
\[ \left|\Pr_{x\gets \zo^k}[D(G(x))=1]-\Pr_{y\gets\zo^{N}}[D(y)=1]\right|\le \varepsilon.\]

Such PRGs are known to follow from assumptions that effectively say: ``There are uniform computations for which nondeterminism (and nonuniformity) does not give arbitrarily large speed-ups.''

In particular, relatively standard (worst-case) derandomization assumptions introduced in the context of showing $\mathsf{NP}=\mathsf{AM}$~\cite{DBLP:journals/cc/KinneMS12,CCC:CU05}: $\mathsf{E}=\mathsf{DTIME}[2^{O(n)}]$ requires exponential size nondeterministic circuits.
This assumption implies PRGs with exponential stretch, but just inverse polynomial (in output length) error.\footnote{There is a black box barrier to getting negligible error PRGs from such assumptions~\cite{CCC:AASY15}.}

\paragraph{Computational depth via nondeterministic hardness.}
It is not difficult to see that such a seed-extending PRG $G$ against nondeterministic hardness must have high computational depth:
\begin{enumerate}
    \item A standard counting argument shows that a random string, $r\gets \zo^N$ given a random string $s\gets \zo^k$ has nearly $N$ conditional $\Kt$ complexity with overwhelming probability.
    
    In particular, there are $2^N$ strings of length $N$ and only $2^\ell$ descriptions of length $\ell$. Hence, a random $r$ cannot be compressed to even $N-\log^2(N)$ bits, except with negligible probability.
    \item Suppose (for the sake of contradiction) that $\Kt(G(x))<\alpha N$ with significant probability. That means that there is a short description of $G(x)$, $\sigma$, that can expand to $G(x)$ in time $t$.
    \item These two facts yield a nondeterministic time $t$ distinguisher: nondeterministically guess a short description $\sigma$ and check that $\sigma$ expands to $y=G(x)$. If $y$ was random, no such $\sigma$ exists with very high probability.
\end{enumerate}

\paragraph{Computational depth with overwhelming probability?} Unfortunately, the noticeable error of the PRGs described above fails to achieve high computational depth with overwhelming probability. Discouragingly, there are even black-box barriers to achieving negligible error in such PRGs~\cite{CCC:AASY15}.

To remedy this, we instead turn to the more nuanced notion of a \emph{multiplicative PRG against nondeterministic distinguishers}~\cite{CCC:AIKS16,CCC:Shaltiel25,DBLP:journals/eccc/DermerS26}. This notion, while failing to achieve negligible distinguishing distance, nonetheless preserves the probability of certain negligible probability events. In particular, an $\varepsilon$-multiplicative PRG against nondeterministic time $t$ has the guarantee that for any nondeterministic time $t$ distinguisher $D$, 
\[ \Pr_{x\gets \zo^k}[D(G(x))=1] < 2 \Pr_{y\gets\zo^{N}}[D(y)=1]+\varepsilon.\]
Such PRGs with $\varepsilon=\negl(N)$ follow from the same assumptions described above~\cite{CCC:Shaltiel25,DBLP:journals/eccc/DermerS26}. Moreover, our simple argument that PRG against nondeterministic distinguishers have computational depth hinged around a particular nondeterministic distinguisher, $D$ (the compression tester), such that $\Pr_y[D(y)=1] = \negl(N)$. It follows that $Pr_x[D(G(x))=1]=\negl(N)$ and hence such $G$ has nearly maximal computational depth with overwhelming probability.

\begin{theorem}[\textit{Informal:} Explicit Incompressibility from Nondeterministic Hardness]
	Assume $\mathsf{E}=\mathsf{DTIME}(2^{O(n)})$ is hard for exponential-size circuits. Then for every polynomial time $t$ and every small constant approximation factor $\rho>0$, there is an explicit function
	\[
	f:\zo^{\poly\log N}\to\zo^N
	\]
	computable in polynomial time and is time $t$-incompressible to even $S=N-O(\rho N \log^2(1/\rho)-\polylog N)$ bits within relative distance $\rho$. Namely, for any $t$ decompression procedure, $D:\zo^S\to\zo^N$, except with negligible probability over random $x$, $f(x)=y$ is $\rho$ far from any string in $D(\zo^S)$.
	
\end{theorem}

At last, we can turn to constructing PoS from such functions with computational depth.

\paragraph{PoS from functions of high computational depth.}
We next show an elementary recipe for constructing a Proof of Space from any function $f$ of high computational depth:
\begin{itemize}
    \item {\bf Initialization phase.} 
    \begin{enumerate}
        \item Verifier samples random $x\gets \zo^k$, and the key to a hash function $h$, and sends both to Prover.
        \item Prover computes $y=f(x)$ and stores to memory. (This phase takes roughly time $T$, the time to evaluate $f$.)
        
        The Prover additionally computes and stores two other items at this time:
        \begin{enumerate}
            \item $r=h^*(y)$, the root of the Merkle tree of $y$ using $h$.
            \item $\pi$, a SNARG, or preprocessing for a succinct interactive argument for the statement ``$r=h^*(f(x))$.''
        \end{enumerate}
        
    \end{enumerate}
    \item {\bf Verification phase.}
    \begin{enumerate}
        \item Verifier samples a set of random challenge indices $i_1,\ldots,i_\ell \in [N]$ and sends them to the Prover.
        \item Prover responds with the root $r$ and Merkle openings of $y_{i_1},\ldots,y_{i_\ell}$, as well as a succinct proof the root of the Merkle tree $r$ is the Merkle root of $y=f(x)$. If the prover does not respond within time $t'< t$ (where $t$ is such that $f(x)$ cannot be compressed relative to time $t$), the verifier rejects.
    \end{enumerate}
\end{itemize}
Intuitively, the succinct argument binds the Prover to $y$ and thus the Verifier can be assured that the openings are correct. There are some subtleties in soundness which we will address later in this overview.

We note that this naive recipe is only possible because $f$ is an explicit function (and does not rely on the random oracle)!

Because the SNARG forces comparison to the verifier chosen $y=f(x)$, the Merkle tree only requires 2nd pre-image resistance and hence can be instantiated from one-way functions~\cite{STOC:AmiRot23}. (However, we do not know appropriate succinct arguments from such assumptions.) Additionally, the succinct argument requires a particular structure which is not shared by all doubly efficient proof systems: the prover can effectively preprocess the whole proving strategy into a succinct state, which can be quickly used when interacting with the verifier. While we do not know how to adapt certain recursive proof systems to this framework, such as~\cite{STOC:ReiRotRot16}, this is a feature of Kilian's classic PCP-based argument and any SNARG (as well as~\cite{STOC:GolKalRot08}, but this imposes undesirable structural limitations).

\paragraph{Warm-up: constant soundness.} To prove the soundness of this framework, we wish to show that if there exists a cheating prover strategy that maintains a small state and yet can convince the verifier to accept with non-negligible probability, then $f$ must not have had high computational depth. In other words, given a cheating prover $\cA$ that, on challenge $x$, produces a small state $\sigma_x$ that convinces the verifier to accept, we want to show how to extract $f(x)$ from the short description $\sigma_x$. However, the key challenge, for this to be a contradiction, our extractor must run in time $t$.

To warm up let's run through how to achieve \emph{constant soundness}, a special case where the runtime restriction is not an issue. A standard Markov argument shows that a constant fraction of $x$ must be ``good'': on input $x$ the cheating prover will (with constant probability) produce a small state $\sigma_x$ that enables it to convince the verifier with constant probability, $\varepsilon$. Moreover, because the argument system and the Merkle tree is negligibly sound (and the cheating prover is polynomial time overall), most PoS soundness violations cannot be due to proofs of false statements. In other words, for almost all ``good'' $x$ with overwhelming probability, any polynomial number proofs $\pi_1,\ldots,\pi_m$ produced by the cheating prover will be valid. Hence, if we query the prover with state $\sigma_x$ on $q$ tuples of indices $(i_1,\ldots,i_\ell)$, we can expect that nearly $\varepsilon q$ of the Prover's responses will provide valid bits of $y$, $(y_{i_1},\ldots,y_{i_\ell})$, with valid openings and proofs that the root was computed correctly. Hence, by properties of the direct product (and tail bounds), we can extract nearly all the bits of $y=f(x)$ with high probability from $\sigma_x$ by making $q= O(|y|/\varepsilon^2) = O(|y|)$ queries in time $O(t'\cdot |y|) < t$ (where $f(x)$ cannot be compressed relative to time $t$).

Note that this reduction consumes many random bits to produce an approximation of $y=f(x)$, which potentially contributes to the description of $y$ (beyond simply $\sigma_x$ and the code of the adversary). However, this randomness can be simply folded into the non-uniform advice using standard methods~\cite{FOCS:Adleman78}.

\paragraph{Negligible soundness.} The approach above fails dramatically if one wants to argue the protocol satisfies negligible soundness. In this case, a violating adversary only convinces the verifier with $N^{-c}$ probability for some arbitrarily large constant $c$. Thus, to expect to see the adversary work correctly even just once (when making random queries) would require $q=N^c$ queries. This means that our reduction would run in time at least $n^c$. To reach a contradiction, this would require $t$, such that $f(x)$ cannot be expanded from a short description in time $t$, to be larger than any polynomial $N^c$. But the $T$, the time to compute $f$, would also have to be super-polynomial and the scheme could not be efficient overall.

Instead, we avoid making random queries and instead take a different approach. We start as before, given a cheating prover $\cA$ that is convincing with inverse polynomial probability, ``good'' instances $x$, which admit descriptions $\sigma_x$ that enable $\cA$ to convince the verifier with probability $N^{-c}$, must be inverse polynomially dense. So it suffices to focus on compressing ``good'' $x$. For simplicity in this introduction, let us assume $\cA$ is deterministic. This means the only randomness in the reduction above is the choices of index tuples made to $\cA$.  We observe that for ``good'' $x$, we do not need to simply blindly sample index bundles at random, but can instead use additional bits to encode ``useful'' queries where $\cA$ is guaranteed to respond. 

Our idea is (roughly) to apply the following sequence of observations:
\begin{enumerate}
    \item If we sample randomly from the set of ``useful'' index bundles/tuples in $[N]^\ell$, we will recover nearly all the bits of $y$ due to the properties of the direct product (i.e.~parallel repetition) with just $q=\Theta(N/\ell)$ queries (with overwhelming probability). If told such a sequence of bundles, the reduction technique above will be very time-efficient.
    \item If we can quickly encode such a sequence of $q=O(N/\ell)$ bundles of $\ell$ indices in $[N]$ using $m\ll\ell$ bits to per bundle, then overall we are only using $m\cdot q \ll \ell \cdot N/\ell = N$ additional bits to describe an approximation of $y$. Thus, if we can do this, the reduction will be efficient in its use of additional description length (in addition to $\sigma_x$), as well as being efficient in time.
    \item Because useful tuples are $N^{-c}$-dense in the set of all queries, we can effectively apply a \emph{coding theorem} to encode these tuples using just $O(\log N)$ bits. In particular, with high probability over a random $\ell\log(N) \times O(\log N)$ dimensional matrix $M$, for many useful tuples $(i_1,\ldots,i_\ell)= M s$ for some $s\in\zo^{O(\log N)}$. Thus, so long as we set $\ell \ggg \log N$, we will arrive at a contradiction.
\end{enumerate}

We do not exactly implement this approach as described. First, while we observe a random matrix suffices\footnote{Actually, we prefer to use a union of such linear maps, as it yields a more compressed description.} to instantiate the approach, we actually replace the linear map above with an explicit $\varepsilon$-Hitting Set Generator (HSG)~\cite{CCC:AIKS16}. This HSG follows from the same assumption as the incompressible function and has a constant size description. Additionally, we use this HSG to hit an entire sequence of ``useful,'' ``covering'' tuples at once (as well as any random bits required for a randomized $\cA$ to be successful on this sequence). Because the event that all the tuples are useful and cover most of $y$ occurs with probability significantly greater than $2^{-N}$ we can obtain a short compression of $y$.

\paragraph{Instantiating the framework.}
We consider two ways of instantiating the succinct argument, with various tradeoffs. (Recall that in addition to a succinct argument, both instantiations additionally rely on the derandomization-style nondeterministic hardness assumption: $\mathsf{E}$ is hard for exponential-size nondeterministic circuits)
\begin{enumerate}
\item Perhaps the most straightforward way, is using a SNARG for P~\cite{FOCS:ChoJaiJin21,EC:HJKS22}. This way very little additional storage overhead is required to store the proof. On the other hand, SNARGs for P is a relatively strong assumption.
\item On the other end of the spectrum, observe that if the argument is implemented with a nonadaptive version of Kilian's argument system then we only need to assume Collision-Resistant Hashing. However, in this case the prover must store a PCP of the computation used to generate $y$ in order to later argue the Merkle root was computed correctly. The length of this PCP is $\tilde{\Omega}(T)=\tilde{\Omega}(|y|)$, which means the honest prover requires much more space than what soundness guarantees in an adversarial prover: $\alpha|y|$ space.
\end{enumerate}

\begin{theorem}[\textit{Informal:} PoS via SNARGs for $\mathsf{P}$]
	Assume $\mathsf{E}$ is hard for exponential-size nondeterministic circuits, one-way functions, and a semi-universal SNARG for $\mathsf{P}$. 
	
	Then, there is a Proof of Space where the honest prover stores $N$ bits and secure against any cheating prover using $(1-o(1))N$ storage.
\end{theorem}

\begin{theorem}[\textit{Informal:} PoS via CRHF]
	Assume $\mathsf{E}$ is hard for exponential-size nondeterministic circuits and collision-resistant hash functions.
	
	Then, there is a constant $c$ such that there is a Proof of Space where the honest prover stores $N$ bits and secure against any cheating prover using $N^{1/c}$ storage.
\end{theorem}

\paragraph{Going forward.}
We note that Goldwasser, Kalai, and Rothblum's classic proof system~\cite{STOC:GolKalRot08} additionally allows the prover to effectively condense a proving strategy into a reasonably succinct state in a preprocessing phase (indeed this exact feature has been used in the context of recent work on nearly optimal derandomization~\cite{FOCS:CheTel21}). Moreover, this proof system is unconditional and requires no additional assumptions, unlike those presented above. On the other hand, GKR requires interaction proportional to the depth of the computation. This means that if we instantiate our framework with GKR, the function $f$ must be computable by a low-depth circuit. This is somewhat diametrically opposed to the 2nd phase timing constraint, and computational depth. So, while this would yield a PoS where the only ``cryptographic'' assumption was one-way functions, the resulting proof system would only be sound against non-parallel provers (who would be forced to sequentially evaluate the function with computational depth, but low circuit depth). Additionally, it is not clear how to construct such a function from the generic derandomization assumptions described above and one would need a more fine-grained tailor made assumption (similar to those used in optimal derandomization).

Whether other doubly-efficient proof systems can be adapted to this setting is an important question that could potentially immediately reduce the assumption footprint.

A second problem, which we do not address here, is composition. In particular, imagine a single cheating prover is interacting with many verifiers. It is desirable that the prover should be required to allocate storage that scales directly in proportion with number of such interactions, i.e.~the prover should \emph{not} be able to amortize the storage costs. To achieve this using our framework it suffices to construct a function $f$ whose computational depth does not amortize ($k$ instances are $k$-times as incompressible). 

Our approach here immediately yields this for small $k$ simply because incompressibility follows from $f$ being a PRG and the direct product of PRGs remains pseudorandom via a hybrid argument. However, because these are not cryptographic PRGs, we can only tolerate distinguishers of bounded size and hence only so many hybrids. Giving a construction that remains incompressible for any polynomial number of parallel instances would yield a PoS which much more desirable security properties.

Finally, while prior pebbling-based ROM approaches to PoS have tended to additionally yield related objects such as Memory-Hard Functions, it is not immediately clear how to adapt the ideas here to get explicit constructions of such objects from relatively simple complexity-theoretic assumptions. New ideas are needed to address this important open question.

%% file: prelim.tex
\section{Preliminaries}

We use $(\out_P, \out_V)\from \langle P(\inp_P), V(\inp_V)\rangle(\inp)$ to denote an interactive protocol between two parties $P$ and $V$ on a shared input $\inp$, local inputs $\inp_P$ and $\inp_V$, with local outputs $\out_P, \out_V$.

\begin{definition}[Proofs of Space~\cite{C:DFKP15}]
    Let $\params=(\lambda, N)$ be the security and space parameters. A Proof of Space (PoS) is an interactive protocol between two parties $P=(P_1, P_2)$ and $V=(V_1, V_2)$, with the following two phases.
    \begin{itemize}
        \item \textbf{Initialization}: $(S, \Phi) \from \langle P_1, V_1\rangle(\params)$. The prover $P_1$ produces a storage advice $S$ of size $N$, and the verifier $V_1$ may output a special symbol $\Phi=\bot$ to abort when interacting with a cheating prover.
        \item \textbf{Execution}: $(\varnothing, 1/0)\from \langle P_2(y), V_2(\Phi) \rangle(\params)$. The verifier outputs a single bit $1/0$ representing either accept or reject.
    \end{itemize}
    
    We require the following properties for a Proof of space.
    \begin{itemize}
        \item \textbf{Completeness}: For any honest prover $P$, \[
        \Pr\left[ \out=1 : \begin{matrix}
            (S, \Phi)\from \langle P_1, V_1 \rangle(\params),\\
            (\varnothing, \out)\from \langle P_2(S), V_2(\Phi) \rangle(\params)
        \end{matrix} \right]=1
        \]
        \item \textbf{$(s,t,\varepsilon)$-Soundness}: For all adversarial provers $\cA=(\cA_1, \cA_2)$ with the size of the advice $|S|\leq s$, and $\cA_2$ using at most $t$ time, we have 
        \[
        \Pr\left[ \out=1:\begin{matrix}
            (S, \Phi)\from \langle \cA_1, V_1 \rangle(\params),\\
            (\varnothing, \out)\from \langle \cA_2(S), V_2(\Phi) \rangle(\params)
        \end{matrix}\right]\leq \varepsilon.
        \]
        \item \textbf{Efficiency}: The prover $P_2$ and the verifiers $V_1$, $V_2$ should be efficient, more specifically, should run in time $\polylog(N)$ and $\poly(\lambda)$. $P_1$ for the Initialization phase should run in time $\poly(N)$. 
    \end{itemize}
\end{definition}

\subsection{Vector Commitment through Merkle Tree}

One building block we will be using for both our constructions is a Vector Commitment scheme instantiated through a Merkle Tree hash. Specifically, we will follow the syntax by~\cite{STOC:AmiRot23}. We reproduce the definition below, with slight adaptations to present the scheme as a tuple of algorithms $(\Commit, \Open, \Verify)$.

\begin{definition}[Merkle Tree]
    Fix $N, \nin, \nout \in \mathbb{N}$ and let $\Sigma$ be a finite alphabet. Given a vector of inputs $\mathbf{y} = (y_1, \dots, y_N) \in \Sigma^N$ and hash functions $h_1, \dots, h_d: \Sigma^{\nin}\to \Sigma^{\nout}$ with $d = \log_{\nin/\nout}(N)$, the Merkle Tree $\Tau=\Tau(\mathbf{y},h_1,\dots, h_d)$ is defined as follows:
    \begin{itemize}
        \item The tree has $d + 1$ layers, indexed from $1$ (root) to $d+1$ (leaves).
        \item There are $L= \lceil N/\nin \rceil$ leaf node \textbf{blocks}. The $j$-th leaf node block, denoted $c_{d+1,j}$, contains the node values $y_{(j-1)\nin + 1}, \dots, y_{j \cdot \nin}$.
        \item For $i \in [d]$, the $i$-th layer is denoted by $w_i$ and is created by applying $h_i$ to each block of the $(i+1)$-st layer: layer $i$ contains $h_i(c_{i+1, j})$ for $j \in \left[ |w_{i+1}| / \nin \right]$, and is further divided into $\nin$-blocks $c_{i,j}$ for $j \in \left[ |w_{i}|/ \nin \right]$ where $|w_i|= \frac{\nout}{\nin} |w_{i+1}|$.
        \item The root is $c = h_1(c_{2,1}) \in \Sigma^{\nout}$.
    \end{itemize}
\end{definition}

Kindly note that the $c_{i,j}$'s are not the nodes of the Merkle Tree, but rather the blocks of the nodes, each containing $\nin$ nodes. For more details, see~\cite{STOC:AmiRot23}.

\begin{definition}[Valid Path]
    Let $N, \nin, \nout, \Sigma, L, d$ and $h_1, \dots, h_d$ be as in the Merkle Tree definition above. Given a root $c \in \Sigma^{\nout}$ and a leaf block index $j \in [L]$, a path $\mathbf{p} = (p_1, \dots, p_{d+1}, c_2, \dots, c_{d+1})$ is called a \emph{valid path} from leaf block $j$ to root $c$ if it satisfies:
    \begin{itemize}
        \item The path starts at the root: $p_1 = c$ ;
        \item $\forall i \in [d]$: $c_{i+1} \in \Sigma^{\nin}$ and contains the block $p_{i+1}$;
        \item $\forall i \in [d]$: $p_i = h_i(c_{i+1})$.
    \end{itemize}
\end{definition}

\begin{definition}[Vector Commitment]
    \label{def:vc}
    Let $\cH$ be a family of hash functions. A Vector Commitment scheme from Merkle tree is a tuple of deterministic algorithms $(\Commit, \Open, \Verify)$:
    \begin{itemize}
        \item $\Commit(\params, \mathbf{h}=(h_1, \dots, h_d), \mathbf{y}=(y_1, \dots, y_N) )\to (c, \st)$: The committing algorithm takes as input the parameters $\params$, a vector of hash functions $h_1, \dots, h_d\in \cH: \Sigma^{\nin}\to \Sigma^{\nout}$, and a vector of inputs $y_1, \dots, y_N\in \Sigma$. It computes the Merkle Tree $\Tau=\Tau(\mathbf{y},h_1,\dots, h_d)$, and outputs the root $c\in \Sigma^\nout$ along with a state $\st=(\mathbf{y}, h_1, \dots, h_d)$.
        \item $\Open(\params, \st, c, i) \to \mathbf{p}$: The opening algorithm takes as input the parameters $\params$, the state $\st$, the commitment $c$, and an index $i\in [N]$. It outputs the opening $\mathbf{p} = (p_1, \dots, p_{d+1}, c_2, \allowbreak \dots, c_{d+1})$ corresponding to the Merkle path from leaf block $i'=\lceil\frac{i}{\nin}\rceil$ to root $c$.
        \item $\Verify(\params, c, i, y_i, \mathbf{p}) \to 1/0$: The verification algorithm takes as input the parameters $\params$, the commitment $c$, an index $i\in [N]$, the claimed value $y_i\in \Sigma$, and a proof $\mathbf{p}$, and outputs accept if and only if $\mathbf{p}$ is a valid Merkle path from the root $c$ to leaf block $i'=\lceil\frac{i}{\nin}\rceil$ and $y_i$ is the $i$-th symbol of the committed vector.
    \end{itemize}
\end{definition}

Instantiating the hash function family $\cH$ with a family of Collision Resistant Hash Functions (CRHFs) yields a stronger \emph{position binding} property, which we need for our second construction that uses a Kilian-style protocol.

\begin{definition}[Position Binding]
    We say the Vector Commitment scheme is $\varepsilon$-position-binding if for a sufficiently large $\nin$, uniformly sampled hash functions $\mathbf{h}\from \cH^d$, and for all non-uniform PPT adversaries $\cA=(\cA_1, \cA_2)$,
    \[
    \Pr\left[  \Verify(\params, c, i, y^*, \mathbf{p}^*)=1 \ \land\  y^* \neq y_{i} \middle| \begin{matrix}  \mathbf{y}\from \cA_1(\params,\mathbf{h})\\(c, \st) \from \Commit(\params, \mathbf{h}, \mathbf{y}) \\ (y^*, \mathbf{p}^*, i) \from \cA_2(\params, \mathbf{h}, \mathbf{y}, c, \st) \end{matrix} \right] < \varepsilon,
    \]
    where probability is taken over the distribution of the hash functions and the random coins of the adversary.
\end{definition}

\begin{lemma}[\cite{C:Merkle87}]
    Assuming Collision Resistant Hash Functions, there exists a position-binding Vector Commitment scheme.
\end{lemma}

Alternatively, we can instantiate the hash family $\cH$ with universal one-way hash functions. This would not give us full position binding, but instead, gives us 2nd pre-image resistance, which is sufficient for our first construction based on SNARGs for $\mathsf{P}$.

\begin{definition}[2-PLOSC~\cite{STOC:AmiRot23}]
    \label{def:2plosc}
    We say the Vector Commitment scheme is an $\varepsilon$-secure Second-Preimage Locally Openable Succinct Commitment (2-PLOSC) if for a sufficiently large $\nin$, uniformly sampled hash functions $\mathbf{h}\from \cH^d$, and for all inputs $\mathbf{y}$, and non-uniform PPT adversaries $\cA$,
    \[
    \Pr\left[  \Verify(\params, c, i, y^*, \mathbf{p}^*)=1 \ \land\  y^* \neq y_{i} \middle| \begin{matrix} (c, \st) \from \Commit(\params, \mathbf{h}, \mathbf{y}) \\ (y^*, \mathbf{p}^*, i) \from \cA(\params, \mathbf{h}, \mathbf{y}, c, \st) \end{matrix} \right] < \varepsilon,
    \]
    where probability is taken over the distribution of the hash functions and the random coins of the adversary.
\end{definition}

\begin{lemma}[\cite{STOC:AmiRot23}]
    \label{lem:2plosc}
    Assuming Universal One-Way Hash Functions (UOWHFs), there exists a 2-PLOSC.
\end{lemma}

And importantly, one way functions imply UOWHFs.

\begin{lemma}[\cite{STOC:Rompel90, eprint:KatKoo05}]
    The existence of one-way functions implies the existence of universal one-way hash functions.
\end{lemma}

\subsection{Non-Interactive (Semi-)Universal Arguments}
For our first construction of Proof of Space, we will use non-interactive (semi-)universal arguments.

\begin{definition}[Universal Language]
    Let $\cL_\cU$ be the language of all tuples $(M, \cf, \cf', T)$ where $M$ is a deterministic Turing machine starting from configuration $\cf$ that ends up in configuration $\cf'$ in $T$ steps. Here, the configuration includes the machine's states and its entire memory. 
\end{definition}

\begin{definition}[Non-Interactive Universal Arguments~\cite{C:BPSS23}]
    An non-interactive argument for $\cL_\cU$ is a tuple of algorithms $\Pi=(\Gen, \Prove, \Verify)$ with the following syntax:
    \begin{itemize}
        \item $\Gen(1^\lambda)\to (\pk, \vk)$: The key generation algorithm takes as input the security parameter $\lambda$ and outputs a prover key $\pk$, and a verifier key $\vk$.
        \item $\Prove(\pk, y)\to \pi$: The prover algorithm takes as input the prover key $\pk$ and an instance $y\in \cL_\cU$, and outputs a proof $\pi$.
        \item $\Verify(\vk, y, \pi)\to 1/0$: The verifier algorithm takes as input the verifier key $\vk$, an instance $y$, and a proof $\pi$, and outputs a single bit representing accept or reject. 
    \end{itemize}
    We require the following properties:
    \begin{itemize}
        \item \textbf{Completeness:} For all $\lambda \in \bbN$ and $y=(M,\cf, \cf', T)\in \cL_\cU$ such that $|y|, T\leq 2^\lambda$, we have \[
        \Pr\left[\Verify(\vk, y, \pi)=1:\begin{matrix}
            (\pk,\vk)\from \Gen(1^\lambda),\\
            \pi\from \Prove(\pk, y)
        \end{matrix}\right]=1.        
        \] 
        \item \textbf{$\varepsilon$-Universal Soundness:} Let $\barT(\lambda)=2^\lambda$. For all $\poly(\lambda)$-sized adversaries $\cA$, we require that for all $\lambda\in \bbN$:
        \[
        \Pr\left[\Verify(\vk, y, \pi)=1 \land y\not \in \cL_M \land T\leq \barT(\lambda):\begin{matrix}
            (\pk,\vk)\from \Gen(1^\lambda),\\
            (y=(M, \cf, \cf', T), \pi)\from\cA(\pk, \vk)
        \end{matrix}\right] \leq \varepsilon.
        \]
        \item \textbf{Efficiency:} We require $\Gen$ to run in time $\poly(\lambda)$, $\Prove$ in time $\poly(\lambda, |y|, T)$, and $\Verify$ in time $\poly(\lambda, |y|)$, and the proof $\pi$ has length $\poly(\lambda)$.
    \end{itemize}
\end{definition}

\begin{remark}
    We can consider non-interactive universal arguments in the Common Reference String (CRS) model, where the prover key $\pk$ and the verifier key $\vk$ can both be computed from the $\crs$. In this setting, we will overload the syntax by omitting the $\Gen$ procedure, and having $\Prove$ and $\Verify$ take $\crs$ as input instead, i.e. $\Prove(\crs, y)\to \pi$ and $\Verify(\crs, y, \pi)\to 1/0$. The security notions remain the same, as in the Soundness experiment, the adversary $\cA$ is allowed access to both $\pk$ and $\vk$, which can be derived from the $\crs$. 
\end{remark}

\begin{definition}[Non-Interactive Semi-Universal Arguments]
    We can also consider non-interactive \emph{semi-universal} arguments, where we replace $\varepsilon$-Universal Soundness with $\varepsilon$-Semi-Universal Soundness, as defined below:

    \begin{itemize}
        \item \textbf{$\varepsilon$-Semi-Universal Soundness:} For every polynomial $\barT=\barT(\lambda)$ and $\poly(\lambda)$-sized adversary $\cA$, we require that for all $\lambda\in \bbN$:
        \[
        \Pr\left[\Verify(\vk, y, \pi)=1 \land y\not \in \cL_M \land T\leq \barT(\lambda):\begin{matrix}
            (\pk,\vk)\from \Gen(1^\lambda),\\
            (y=(M, \cf, \cf', T), \pi)\from\cA(\pk, \vk)
        \end{matrix}\right] \leq \varepsilon.
        \]
    \end{itemize}
\end{definition}

\begin{remark}
    We can further relax the Soundness properties by giving the adversary $\cA$ only the $\pk$, not the $\vk$. In this case, we achieve a weaker notion of soundness, and we say the scheme is \emph{privately verifiable}.
\end{remark}

\subsection{Kilian's Protocol}
\label{sec:kilian-prelim}
For our second construction of Proof of Space, we will use interactive arguments, e.g.\ Kilian's protocol.

Kilian's protocol is based on the idea of using Probabilistically Checkable Proofs (PCPs) to argue that a committed string encodes a valid transcript for a NP (or more general) relation.

\begin{definition}[Probabilistically Checkable Proof (PCP)]
    \label{def:pcp}
    Let $R \subseteq \{0,1\}^* \times \{0,1\}^*$ be a relation. A \emph{probabilistically checkable proof} (PCP) system for $R$ is a pair of PPT algorithms $\PCP=(P,V)$ with the following syntax:
    \begin{itemize}
        \item $P(x,w) \to \Pi$: On common input $x$ and prover input $w$ (a witness for $x$), output a PCP string $\Pi=(\Pi_1,\dots,\Pi_\ell) \in \Sigma^\ell$.
        \item $V(x;\rho)$: On input $x$ and randomness $\rho \in \zo^r$, the verifier makes $q$ oracle queries to $\Pi$ at indices in a set $\cQ(\rho)=\{j_1,\dots,j_q\}\subseteq[\ell]$ (possibly with repetitions). We write $V^\Pi(x;\rho)$ for the verifier that receives oracle answers $\Pi[j]$ for each $j\in \cQ(\rho)$.
    \end{itemize}
    For a query set $\cQ \subseteq [\ell]$ and answers $\ans \in \Sigma^{|\cQ|}$, we write $V^{[\cQ,\ans]}(x;\rho)$ for the verifier run on randomness $\rho$ when each query $j\in \cQ(\rho)$ is answered with $\ans[j]$; if $\cQ(\rho)\not\subseteq \cQ$, then $V^{[\cQ,\ans]}(x;\rho)$ outputs $0$.
    We require:
    \begin{itemize}
        \item \textbf{Completeness:} If $(x,w)\in R$, then
        \[
        \Pr\left[V^{\Pi}(x;\rho)=1 : \rho\gets \zo^r,\ \Pi\from P(x,w)\right]=1.
        \]
        \item \textbf{$\varepsilon_{\mathsf{pcp}}$-Soundness:} For all $x$ with no $w$ such that $(x,w)\in R$, and all $\tilde{\Pi}\in \Sigma^\ell$,
        \[
        \Pr\left[V^{\tilde{\Pi}}(x;\rho)=1 : \rho\gets \zo^r\right] \leq \varepsilon_{\mathsf{pcp}}.
        \]
        \item \textbf{Efficiency:} $P$ runs in time $\poly(|x|,|w|)$; $V$ runs in time $\poly(|x|,r)$ and makes $q$ queries.
    \end{itemize}
    We call $\ell$ the \emph{proof length}, $q$ the \emph{query complexity}, and $r$ the \emph{verifier randomness length}.
\end{definition}

Kilian's argument~\cite{STOC:Kilian92} combines a PCP with a vector commitment to the PCP string: the verifier checks PCP soundness on the queried symbols, and the commitment binds those symbols to a single proof $\Pi$.

\begin{definition}[Kilian commit-and-prove protocol]
    \label{def:kilian}
    Let $R$ be a relation and let $\PCP=(P,V)$ be a PCP for $R$ with proof length $\ell$.
    Let $\MERKLE$ be a position-binding vector commitment scheme instantiated with Collision Resistant Hash Functions.
    The \emph{Kilian commit-and-prove protocol} for $R$ is the following two-phase protocol on input $x$:
    \begin{itemize}
        \item \textbf{Commitment (preprocessing):} The prover with witness $w$ computes $\Pi\from P(x,w)$ and $(c,\st)\from \MERKLE.\Commit(\params,\Pi)$, and sends $c$ to the verifier.
        \item \textbf{Argument (one round):} The verifier sends $\rho\gets \zo^r$; the prover computes $\cQ$ from $V(x;\rho)$, $\pf\from \MERKLE.\Open(\params,\st,c,\cQ)$, sets $\ans=\Pi[\cQ]$, and sends $(\cQ,\ans,\pf)$; the verifier accepts iff $\MERKLE.\Verify(\params,c,\cQ,\ans,\pf)=1$ and $V^{[\cQ,\ans]}(x;\rho)=1$.
    \end{itemize}
\end{definition}

\begin{definition}[Computational soundness of Kilian]
    \label{def:kilian-sound}
    The protocol in Definition~\ref{def:kilian} has \emph{computational soundness error} $\varepsilon_{\mathsf{kil}}$ if for all PPT provers $\cP^*$ and all $x$ with no witness $w$ such that $(x,w)\in R$,
    \[
    \Pr\left[\text{verifier accepts on input }x\right] \leq \varepsilon_{\mathsf{kil}},
    \]
    where probability is over the verifier's randomness $\rho$ and the coins of $\cP^*$.
    If $\MERKLE$ is $\varepsilon$-secure and $\PCP$ has soundness $\varepsilon_{\mathsf{pcp}}$, then $\varepsilon_{\mathsf{kil}} \leq \varepsilon_{\mathsf{pcp}} + q\cdot \varepsilon$ for query complexity $q=|\cQ|$.
    With $\MERKLE$ a position-binding vector commitment scheme and negligible $\varepsilon_{\mathsf{pcp}}$, one obtains $\varepsilon_{\mathsf{kil}}=\negl(\lambda)$~\cite{STOC:Kilian92}.
\end{definition}

\subsection{Error Correcting Codes}
\begin{definition}[Error Correcting Codes]
    An Error Correcting Code (ECC) for message space $\{0,1\}^n$ and encoding space $\{0,1\}^{p(n)}$ is a tuple of polynomial time algorithms $\ECC=(\En, \De)$ with the following syntax:
    \begin{itemize}
        \item $\En(m)\to c$: The encoding procedure takes as input a message $m\in \{0,1\}^n$ and outputs its corresponding codeword $c\in\{0,1\}^{p(n)}$.
        \item $\De(c)\to m$: The decoding procedure takes a codeword $c\in\{0,1\}^{p(n)}$ and outputs the corresponding message $m\in \{0,1\}^n$.
    \end{itemize}

    Correctness requires that for all $m\in \{0,1\}^n$, $\De(\En(m))=m$.

    An $\ECC$ has relative distance $\alpha>0$ if for all message lengths $n\in \bbN$ and any $m_0\neq m_1\in\{0,1\}^n$, we have 
    \[
    \Delta(\En(m_0),\En(m_1))\geq \alpha p,
    \]
    where $\Delta$ denotes the Hamming distance.
    
    A relative distance of $\alpha$ means the $\ECC$ is capable of correcting up to $\frac{1}{2}\alpha p$ errors, i.e., for all $m\in \{0,1\}^n$ and $c'\in \{0,1\}^{p(n)}$, if $\Delta(\En(m), c')< \frac{1}{2}\alpha p$, then $\De(c')=m$.
    
    \ignore{The main property we will utilize is \emph{Error Correction}.
    \begin{itemize}
        \item \textbf{$d$-Error Correction}: For all $r$ s.t. $\Delta(r, \En(m))\leq d$ for some $m$, we have $\De(r)=m$, where $\Delta$ denotes the Hamming distance.
    \end{itemize}}
\end{definition}

It is known binary error correcting codes with relative distance $\approx \frac{1}{2}$ exists with a linear blowup in encoding length.

\begin{lemma}[\cite{STOC:GurSud00}]
    For all $\epsilon>0$, there exists an error correcting code $\ECC$ with relative distance $\frac{1}{2}-\varepsilon$ and with $p=p(n)=O_{\varepsilon}(n)$.
\end{lemma}

\ignore{
\begin{definition}[Error Correcting Code]
    Let $\Sigma$ be a finite alphabet, and let $C$ be a subset of $\Sigma^n$ (that is, a set of length $n$ words in the alphabet). Then, we say $C$ is a code over $\Sigma$, and the minimum distance of $C$ is defined as the minimum Hamming distance between elements of $C$.

    Further, if we let $\mathcal{M}$ be some message space, then we call $\mathbf{A}:\cM\rightarrow C$ an encoding of $\mathcal{M}$ for code $C$. When appropriate in context, we may also call $\mathbf{A}$ a code.
\end{definition}}

\subsection{Hitting Set Generator (HSG)}
\begin{definition}[Hitting Set Generator (HSG)]
    Let $\cC$ be a class of boolean functions $C : \{0, 1\}^n \to \{0, 1\}$.
    A function $G : \zo^r \to \zo^n$ is an $\varepsilon$-HSG for $\cC$ if for every $C \in \cC$ such that $\Pr[C(U_n) = 1] > \varepsilon$, there exists $x \in \zo^r$ such that $C(G(x)) = 1$.
\end{definition}

\begin{lemma}[Explicit HSG~\cite{CCC:AIKS16}]
    \label{lem:hsg}
    If $\mathsf{E}=\mathsf{DTIME}(2^{O(n)})$  is hard for exponential size nondeterministic circuits,  then for every constant $b > 1$, there exists a constant $c > 1$ such that for every sufficiently large $n$, there is a function $G : \zo^r\to \zo^n$ that is an $\varepsilon$-HSG for size $n^b$ circuits, with $r = \log (1/\varepsilon) + c \log n$. Furthermore, $G$ is computable in time $\poly(n^b)$.
\end{lemma}

\begin{lemma}[Monte Carlo HSG]\label{lem:mc-hsg}
    There is a family $\cG=\{G:[r]\times(\zo^d\setminus\{0^n\})\to\zo^n\}$ such that any $G\gets \cG$ can be sampled with $r(2n+d-1)$ random bits and computed in time $O(rnd)$ and for any set $S$ with density $\epsilon$,
    \[
        \Pr_{G\gets\cG}[\forall x\in[r]\times(\zo^d\setminus\{0^n\}), G(x)\notin S]\le \left(\frac{1}{(2^d-1)\epsilon}\right)^r.
    \]
    Hence, for any $\delta\in(0,1)$ and $\cS=\{S_1,\ldots,S_M\}$ where each $S_i\in\cS$ has density $\ge \epsilon$, if we take $d\ge \log(1/\epsilon)+2$ and $r=\log(M)+\log(1/\delta)$ then
    \[
        \Pr_{G\gets \cG}[\exists S\in \cS, \forall x, G(x)\notin S]\le \delta.
    \]
\end{lemma}
\begin{proof}[Sketch]
    Take $G$ to be the union of $r$ linear maps defined by Toeplitz matrices. The first coordinate in $[r]$ defines which linear map to apply the second input to. The Toeplitz family is pairwise independent (provided we exclude $0^n$), hence Chebyshev bound will give that a single random map fails to hit any particular $S$ with probability at most $p=1/(2^d-1)\epsilon$. The probability of this event under $r$ independent trials is at most $p^r$. A specific choice of parameters and union bound over $\cS$ yields the last line.
\end{proof}

%% file: ktcomp.tex
\section{Hard Functions with High Time-Bounded Kolmogorov Complexity}

Our constructions will rely on a central tool, a function $f:\zo^{\lin} \to \zo^{\lout}$ that can be computed in time $T$, but has a high conditional time-bounded Kolmogorov complexity $\Kt(f(x) \mid x)$ for time $t\ll T$. Specifically, we will consider three different variants of time-bounded Kolmogorov complexity.

We first start with the standard definition of time-bounded Kolmogorov complexity. Towards the definition, we will use the universal Turing machine $U$ that can emulate any Turing machine with polynomial overhead. The universal Turing machine $U$ takes as input (the description of) a program $\Pi=(M, w)\in \{0,1\}^*$, where $M$ is a Turing machine and $w\in \{0,1\}^*$ is an input to $M$. We use $U(M, w, 1^t)$ to denote the output of $M(w)$ in $t$ steps when emulated on $U$.

\begin{definition}[Time-Bounded Kolmogorov Complexity]
    Let $U$ be a universal Turing machine. The \emph{$t$-time-bounded Kolmogorov complexity} of s string $y$ conditioned on $x$, denoted $\Kt(y \mid x)$, is defined as
    \[
    \Kt(y \mid x) = \min_{M \in \{0,1\}^*} \left\{ |M| : U(M, x, 1^t) = y \right\}.
    \]
    That is, $M$ is a program such that emulating $M$ on input $x$ outputs $y$ in $t$ steps.
    
    We say $f:\zo^{\lin} \to \zo^{\lout}$ has \emph{$t$-time Kolmogorov complexity at least $s$} if
    \[
    \Pr_{x\from \cU_{\lin}}\left[\Kt(f(x) \mid x) > s\right] > 1-\varepsilon,
    \]
    for a negligible $\varepsilon =\negl(\lout)> 0$, where $\cU_{\lin}$ denotes the uniform distribution over $\lin$-bit strings\footnote{Here we assume the input distribution for $f$ is the uniform distribution. One could also define the notion analogously for arbitrary input distributions.}.
\end{definition}

Next, we strengthen the notion to account for probabilistic Turing machines. We will let $U$ be a universal Turing machine that emulates \emph{probabilistic} Turing machines with polynomial overhead. The universal Turing machine $U$ takes as input (the description of) a program $\Pi=(M, w)\in \{0,1\}^*$, where $M$ is a Turing machine and $w\in \{0,1\}^*$ is an input to $M$, as well as a randomness $r\in \{0,1\}^*$. We use $U(M, w, r, 1^t)$ to denote the output of $M(w;r)$ in $t$ steps when emulated on $U$, where $r$ is fed to $M$ as a random tape.

\begin{definition}[Time-Bounded Probablistic Kolmogorov Complexity~\cite{CCC:GKLO22}]
    Let $U$ be a universal Turing machine that emulates probabilistic Turing machines. The \emph{$t$-time-bounded probabilistic Kolmogorov complexity} of a string $y$ conditioned on $x$, denoted $\pKt(y \mid x)$, is defined as

    \[
    \pKt(y \mid x) = \min\left\{ k\in \mathbb{N}\  \middle| \  \Pr_{r\from \cU_t}\left[\exists M \in \{0,1\}^k: U(M, x, r, 1^t)=y\right]\geq \frac{2}{3}\right\},
    \]
    where $\cU_t$ denotes the uniform distribution over $t$-bit strings.
    
    We say $f:\zo^{\lin} \to \zo^{\lout}$ has \emph{$t$-time probabilistic Kolmogorov complexity at least $s$} if
    \[
    \Pr_{x\from \cU_{\lin}}\left[\pKt(f(x) \mid x) > s\right] > 1-\varepsilon,
    \]
    for a negligible $\varepsilon =\negl(\lout)> 0$, where $\cU_{\lin}$ denotes the uniform distribution over $\lin$-bit strings.
\end{definition}

One way to view the definition above is that the randomness $r$ is provided to the Turing machine $M$ as uniform advice. Naturally, one can consider the stronger notion where \emph{an arbitrary non-uniform advice} is provided to the Turing machine $M$. However, such a notion no longer makes sense to measure specific strings (one can simply write the string in the advice). Though it does make sense as a measure of the incompressibility of a distribution over strings.

\begin{definition}[Kolmogorov Complexity with non-uniform Advice]
    \label{def:aKt}
    Let $U$ be a universal Turing machine. The \emph{$t$-time-bounded Kolmogorov complexity with non-uniform advice $z$} of a string $y$ conditioned on $x$, denoted $\aKt_z(y \mid x)$, is defined as
    \[
      \aKt_z(y \mid x) = \min_{M \in \{0,1\}^*} \left\{ |M| : U(M, (x, z), 1^t) = y \right\}.
    \]

    That is, $M$ is a program that outputs $y$ given input $x$ and a non-uniform advice $z$ in time $t$. Notice that trivially, $\aKt_y(y \mid x) = O(1)$, when the non-uniform advice happens to be the output string $y$ itself (or more generally, when the non-uniform advice can depend on the output string $y$).
\end{definition}

\begin{definition}[$t$-time Incompressibility]
    We say a function $f:\zo^{\lin} \to \zo^{\lout}$ is \emph{$(t, s)$-incompressible} if for any non-uniform advice $z \in \{0,1\}^*$, we have
    \[
    \Pr_{x\from \cU_{\lin}}\left[\aKt_z(f(x) \mid x) > s\right] > 1-\varepsilon,
    \]
    for a negligible $\varepsilon =\negl(\lout)> 0$, where $\cU_{\lin}$ denotes the uniform distribution over $\lin$-bit strings.
\end{definition}

Notice the order of the quantifiers in the definition above. The non-uniform advice $z$ cannot depend on the input $x$, otherwise the definition is trivial as pointed out in Definition~\ref{def:aKt}. Instead, a ``useful'' advice $z$ needs to work for all inputs $x$.

\subsection{Approximate Incompressibility}

To utilize these hard functions with high $\Kt$ complexity, we compose a function of certain $\Kt$ complexity with an \emph{Error Correcting Code} to obtain a function with high \emph{Approximate} $\Kt$ complexity.

\begin{definition}[Approximate Kolmogorov Complexity]
    Let $f:\zo^{\lin} \to \zo^{\lout}$ be a function. We say $f$ has $\delta$-approximate $t$-time Kolmogorov complexity $s$ if
    
    \[\Pr_{x\from \cU_{\lin}}\left[\exists y': \Delta(y', f(x)) < \delta \cdot |f(x)| \ \wedge \  \Kt(y' \mid x) \leq s\right] < \varepsilon,\]
    for a negligible $\varepsilon =\negl(\lout)> 0$, where $\cU_{\lin}$ denotes the uniform distribution over $\lin$-bit strings.

    That is, with overwhelming probability over choice of $x$, we have $\Kt(y' \mid x) > s$ for all $y'$ such that the Hamming distance $\Delta(y', f(x)) < \delta \cdot |f(x)|$. We denote this as \[\Kt_{\approx\delta}(f(x)\mid x) > s.
    \]
\end{definition}

Analogously, we can define approximate $\pKt$ and $(t,s)$-incompressibility.

\begin{definition}[Approximate Probabilistic Kolmogorov Complexity]
    Let $f:\zo^{\lin} \to \zo^{\lout}$ be a function. We say $f$ has $\delta$-approximate $t$-time probabilistic Kolmogorov complexity $s$ if
    
    \[\Pr_{x\from \cU_{\lin}}\left[\exists y': \Delta(y', f(x)) < \delta \cdot |f(x)| \ \wedge\  \pKt(y' \mid x) \leq s\right] < \varepsilon,\]
    for a negligible $\varepsilon =\negl(\lout)> 0$, where $\cU_{\lin}$ denotes the uniform distribution over $\lin$-bit strings.

    That is, with overwhelming probability over choice of $x$, we have $\pKt(y' \mid x) > s$ for all $y'$ such that the Hamming distance $\Delta(y', f(x)) < \delta \cdot |f(x)|$. We denote this as \[\pKt_{\approx\delta}(f(x)\mid x) > s.
    \]
\end{definition}

\begin{definition}[Approximate $t$-time Incompressibility]
    Let $f:\zo^{\lin} \to \zo^{\lout}$ be a function. We say $f$ is \emph{$\delta$-approximate $(t,s)$-incompressible} if for any non-uniform advice $z\in\{0,1\}^*$, we have
     \[\Pr_{x\from \cU_{\lin}}\left[\exists y': \Delta(y', f(x)) < \delta \cdot |f(x)| \ \wedge\  \aKt_z(y' \mid x) \leq s\right] < \varepsilon,\]
    for a negligible $\varepsilon =\negl(\lout)> 0$, where $\cU_{\lin}$ denotes the uniform distribution over $\lin$-bit strings.

    That is, for any non-uniform advice $z$, with overwhelming probability over choice of $x$, we have \[\aKt_z(y' \mid x) > s\] for all $y'$ such that the Hamming distance $\Delta(y', f(x)) < \delta \cdot |f(x)|$.
\end{definition}

\subsection{Approximate Incompressibility via List-Decoding}
We observe that approximate incompressibility follows generically by composing with an efficient list-decodable code. We actually do not rely on this result and instead use a direct instantiation that is more efficient.

\begin{theorem}[\cite{DBLP:conf/approx/GuruswamiR07}]
    For every field $\bbF_q$, reals $\delta\in(0,1),r\in(0,1-H_q(\delta)],\epsilon>0$ and integer $s\ge 1$, there exists linear codes $C$ over $\bbF_q$ of block length $n$ that can be list-decoded from $\delta-\epsilon$-relative errors with list size $L(n)=(n/\epsilon^2)^{O(s\epsilon^{-3}\delta/(H_q^{-1}(1-r)-\delta))}$ and rate $R=r-\frac{r}{s}\sum_{i=0}^{s-1}\frac{\delta}{H^{-1}_q(1-r+ri/s)}$. Moreover, $C$ can be constructed in time $(n/\epsilon^2)^{O(s/(\epsilon^6r\delta))}$ and list-decoded in time polynomial in $n$.
\end{theorem}

As noted by Guruswami and Rudra~\cite{DBLP:conf/approx/GuruswamiR07}, this meets the Blokh--Zyablov bound, $R_{BZ}(\rho)$ where $\rho$ is the relative distance of the code, as $s\to \infty$:

\begin{corollary}[\cite{DBLP:conf/approx/GuruswamiR07}]\label{cor:GR-BZ}
    For every \(0<\rho<\tfrac12\) and every \(0<\eta<\rho\), there exists an explicit family of binary linear codes that is list-decodable up to $\rho-\eta$ relative distance with list size $n^{O_{\rho,\eta}(1)}$ and rate at least $R_{BZ}(\rho)-\eta$ where $n$ is the block length and
    \[
        R_{BZ}(\rho) := 1-H_2(\rho)-\rho\int_0^{1-H_2(\rho)}\frac{dx}{H_2^{-1}(1-x)}.
    \]
    Moreover encoding and list-decoding run in $n^{O_{\rho,\eta}(1)}$ time.
\end{corollary}

We rephrase\footnote{
    \[
    1-R_{BZ}(\rho)=H_2(\rho)+\rho\int_0^{1-H_2(\rho)}\frac{dx}{H_2^{-1}(1-x)}
    \]
    
    If we define $u=H_2^{-1}(1-x)$ for the purposes of substitution, then $x=1-H_2(u)$ and $$dx=-H'_2(u)du=-\log_2\frac{1-u}{u} du.$$ Additionally, note that as $x$ goes from 0 to $1-H_2(\rho)$, $u$ goes from $\tfrac12$ to $\rho$.
Hence,
\[
    \int_0^{1-H_2(\rho)}\frac{dx}{H_2^{-1}(1-x)} = \int_{\rho}^{1/2}\frac{\log((1-u)/u)}{u}du.
\]
For $0<u\le1/2$, $\log\frac{1-u}{u} \le \log\frac{1}{u}$. So for $\rho\in(0,1/2)$,
\[
\int_\rho^{1/2} \frac{\log((1-u)/u)}{u}du \le \int_{\rho}^{1/2}\frac{\log(1/u)}{u}du = 1/2(\log^2(1/\rho)-\log^2(2))=\log^2(1/\rho)-1/2.
\]
Additionally, $H_2(\rho) = \rho\log(1/\rho)+(1-\rho)\log(1/(1-\rho))\le \rho\log(1/\rho) + (1-\rho)\frac{\rho}{1-\rho} \le \rho(\log(1/\rho)+1)$. And for $\rho<1/4$, $\log(1/\rho)\le \frac{1}{2}\log^2(1/\rho)$.
So,
\[R_{BZ}(\rho)\ge 1-\rho(\log^2(1/\rho)+1/2).\]} this in the following convenient form (explicit here means that both encoding and list-decoding run in time polynomial in the block length):
\begin{corollary}[Simple High-rate Formulation of~\cite{DBLP:conf/approx/GuruswamiR07}]\label{cor:ECC}
    There exists an absolute constant $c>0$ such that for every sufficiently small $0<\rho<\tfrac14$ and every $0<\eta<\rho$, there is an explicit family of linear binary codes with list size $n^{O_{\rho,\eta}(1)}$ and rate at least
    \[
        1-\rho(\log^2(\frac{1}{\rho})+1/2)-\eta
    \]
    that is explicitly list-decodable up to relative distance $\rho-\eta$.
\end{corollary}

\begin{theorem}
    \label{thm:kt-ecc}
    Let $\params=(\lambda, N, f)$ where $\lambda$ is the security parameter and $N\gg \lambda$ is the space parameter. Let $T(\lambda,N), t(\lambda, N)=\poly(\lambda, N)$, and $\lin(\lambda, N)=\poly(\lambda, \log N)$, and $\lout(\lambda, N)=N/\poly(\lambda)$.

    Let $f:\{0,1\}^{\lin} \to \{0,1\}^{\lout}$ be a function computable in time $T$ such that one of the following holds:
    \begin{enumerate}[(a)]
        \item for $(1-\varepsilon)$ fraction of $x$, $\Kt(f(x)|x)>s(\lambda, N)$;
        \item for $(1-\varepsilon)$ fraction of $x$, $\pKt(f(x)|x)>s(\lambda, N)$;
        \item $f$ is $(t, s(\lambda, N))$-incompressible (w.r.t. an overwhelming $(1-\varepsilon)$ fraction of $x$).
    \end{enumerate} Let $\ECC=(\En,\De)$ be an error correcting code list-decodable from relative distance $\rho-\eta$ with block length $n=O(N)$, list size $\poly(n)=\poly(N)$ and decoding time $t_{\De}=\poly(n)=\poly(N)$. Define $f'(x)$ to be the following function \[
    f'(x)=\ECC.\En(f(x)).
    \]
    Then for $\delta=\rho-\eta$, we have, respectively,
    \begin{enumerate}[(a)]
        \item $\mathsf{K}_{\approx\delta}^{t-t_{\De}}(f'(x)\mid x) > s(\lambda, N)-O(\log N)$;
        \item $\mathsf{pK}_{\approx\delta}^{t-t_{\De}}(f'(x)\mid x) > s(\lambda, N)-O(\log N)$;
        \item $f'$ is $\delta$-approximate $(t-t_{\De}, s(\lambda, N)-O(\log N))$-incompressible.
    \end{enumerate}
\end{theorem}

\begin{proof}
    We prove the theorem for case (a), cases (b) and (c) follow analogously.

    We prove case (a) via a reduction to the property of $f$. We show that if there exists a string $y$ with $K^{t-t_{\De}}(y \mid x) \leq s(\lambda, N) - O(\log N)$ and $\Delta(y, f'(x)) < \delta \cdot |f'(x)|$, then we can use $y$ to efficiently compute $f(x)$, contradicting the property of $f$.

    Let $c$ be a constant such that the list size of $\ECC$ is at most $N^{c}$ (such $c$ exists since the list size is $\poly(n)$ and $n=O(N)$). Notice that since $K^{t-t_{\De}}(y \mid x) \leq s(\lambda, N)-c\log N-O(1)$, this means there exists a small state $\st$ with $|\st| \leq s(\lambda, N)-c\log N-O(1)$ and an algorithm $\cA$ such that $y = \cA(x, \st)$ and $\cA$ runs in time $t-t_{\De}$.

    We can construct an algorithm $\cB$ that uses $\cA$ as a subroutine to compute $f(x)$ as follows:

    \begin{mdframed}
        $\cB(x, \st, i):$
        \begin{enumerate}
            \item Compute $y = \cA(x, \st)$.
            \item Compute $f(x) = \ECC.\De(y,i)$ (Output the $i$th item in the list).
        \end{enumerate}
    \end{mdframed}

    It should be clear to see that $\cB$ runs in time $t$ and takes as input a state of size at most $|\st| + |i| + O(1) = s(\lambda,N)$.

    Correctness follows from correctness of list-decoding under the list-decoding radius.
\end{proof}

\section{Computational Depth from Hardness against Nondeterminism}

\begin{definition}
    We say a function $G:\zo^k\to\zo^n$ is an $(\varepsilon,\delta)$-multiplicative PRG for a computational class $\mathcal{C}$ if for any family of circuits $C\in\mathcal{C}$,
    \[\Pr[C(G(\cU_k))=1] \le e^\varepsilon\Pr[C(\cU_n)=1]+\delta\]
    where $\cU_k$ and $\cU_n$ denote random variables uniformly distributed over $k$-bit and $n$-bit strings, respectively.

\end{definition}

Very recently, Dermer and Shaltiel~\cite{DBLP:journals/eccc/DermerS26}, showed that the classic Shaltiel Umans construction~\cite{FOCS:ShaUma01} in conjunction with techniques developed in the context of randomness extraction~\cite{FOCS:TreVad00,FOCS:BGDM23,CCC:Shaltiel25} can be adapted to give such a PRG against the class of size $s$ nondeterministic circuits, with essentially optimal parameters.
\begin{theorem}[\cite{DBLP:journals/eccc/DermerS26}]\label{thm:PRG}
    If $\mathsf{E}$ is hard for exponential-size nondeterministic circuits (i.e.~$\exists L\in\mathsf{E}$ and constant $\beta>0$ such that any family of nondeterministic circuits of size $2^{\beta n}$ fails to decide $L$ on all but a finite set of input lengths), then for every sufficiently large $s$ and every $\delta \in [2^{-s},1/s]$ there is a $(1/s,\delta)$-multiplicative PRG,
    \[ G:\zo^{r=O(\log 1/\delta)}\to \zo^s\]
    for size $s$ nondeterministic circuits that is computable in time $\poly(s)$.
\end{theorem}

A distinguisher, $D$, is considered to be a $\gamma$-limited-nondeterministic $T$-time (non-uniform) computation if there exists a $T$-time (non-uniform) distinguisher $V$ such that $D(x)=1 \iff \exists w\in\zo^\gamma, V(x,w)=1$. In other words, such a distinguisher can use nondeterministic witness of length up to $\gamma$. Any limited nondeterministic computation is obviously a nondeterministic computation.

\begin{lemma}
    Let $S(t)$ denote the time to simulate a machine for $t$ steps. Let $\delta>0$.
    Let $G:\zo^k\to\zo^n$ be an $(\ln(2),\delta)$-multiplicative PRG for time-$(S(t)+O(n))$ $(n-\log(1/\delta))$-nondeterministic (non-uniform) computations. Then, for any advice $z\in\zo^*$
    \[ \Pr_{x\gets \cU_k}[\aKt_z(G(x)\mid x)\le n-\log(1/\delta)-k] \le 3\delta \]
\end{lemma}
We note that if one starts with a \emph{seed-extending} multiplicative PRG for nondeterminsitic circuits, then one will not suffer the loss of $k$ in the compression bound above. A seed-extending PRG remains pseudorandom even when given the seed (i.e.~$(X,G(X))\approx \cU$ for uniform seed $X$). In particular, applying Nisan Wigderson~\cite{DBLP:conf/focs/NisanW88} to Shaltiel's multiplicative (low-stretch) PRG~\cite{CCC:Shaltiel25} indeed yields such a PRG, albeit with worse seed length than the recent construction of Dermer and Shaltiel~\cite{DBLP:journals/eccc/DermerS26}. This latter construction~\cite{DBLP:journals/eccc/DermerS26} can perhaps be made seed extending (as was the case for Shaltiel and Uman's PRG~\cite{FOCS:ShaUma01} which it is based on), however we have not yet verified this.

\begin{proof}
    Consider the nondeterministic distinguisher $D_z$ that on input $y\in\zo^{n}$ accepts if and only if there exists a witness $(\sigma,x)\in\zo^{(n-\log(1/\delta))\times k}$ such that the output of simulating $\sigma$ with advice $z$ for $t$ steps on input $x$ yields $y$.

    By our assumption about the model, simulating $\sigma$ for $t$ steps takes time $S(t)$ so the distinguisher runs in time $S(t)$.

    First, note that by a counting argument it is the case that for any $m\in \mathbb{N}$ and $z\in\zo^*$, $\Pr[\exists x\in\zo^k, \aKt_z(\cU_n|x)\le m] \le 2^{m+k-n}$. So, $\Pr[\exists x\in\zo^k, \aKt_z(\cU_n|x)\le n-\log(1/\delta)-k]\le \delta$ and hence $\Pr[D_z(\cU_{n})=1]\le 3\delta$.

    On the other hand, assume towards contradiction that there exists $z$ such that $\Pr_x[\aKt_z(G(x)\mid x)\le n-\log(1/\delta)]>3\delta$. Then, $\Pr[D(G(\cU_k))=1]>3\delta$. But because $G$ is multiplicative PRG, we reach the following contradiction
    \[\Pr[D_z(G(\cU_k))=1]\le  2 \Pr[D_z(\cU_{n})=1]+\delta \le  3\delta.\]
\end{proof}

We observe that the conclusion can actually be strengthened significantly by considering the nondeterministic distinguisher that accepts nearby strings.
\begin{lemma}
	Let $S(t)$ denote the time to simulate a machine for $t$ steps. Let $\rho\in(0,1/2)$ and $\delta>0$.  Let $G:\zo^k\to\zo^n$ be an $(\ln(2),\delta)$-multiplicative PRG for time-$(S(t)+O(n))$ nondeterministic (non-uniform) computations. Then, for any advice $z\in\zo^*$
    \[ \Pr_{x\gets \cU_k}[\exists \hat{y}\in\zo^n : \Delta(\hat{y},G(x)\le \rho n\wedge \aKt_z(\hat{y}\mid x)\le n-\log(1/\delta)-k-O(1)] \le 3\delta. \]

\end{lemma}
\begin{proof}
	Consider the nondeterministic distinguisher $D_z$ that on input $y\in\zo^{n}$ accepts if and only if there exists a witness $(\sigma,x)\in\zo^{(n-\log(1/\delta))\times k}$ such that the output of simulating $\sigma$ with advice $z$ for $t$ steps on input $x$ yields $\hat{y}$ with $\Delta(\hat{y},y)\le \rho n$.
	Next observe that any hamming ball of radius $\rho n$ has size at most $2^{H_2(\rho)n}$. So by a counting argument we have that for any $z,m$,
	\[
		\Pr_{y\gets\cU_{n}}[D_z(y)=1]\le \sum_x\Pr_y[D_z(x,y)=1]\le 2^{k+H_2(\rho)n}.
	\]
	 The remainder of the proof proceeds similarly to the previous.
\end{proof}

By simply truncating the PRG from Theorem~\ref{thm:PRG} and applying the lemma above, we get an explicit function with essentially maximal computational depth.
\begin{corollary}
    There are constants $\alpha,\beta,\gamma$ such that the following holds.
    If $\mathsf{E}$ is hard for exponential size nondeterministic circuits, then for any $t$ such that $t=\poly(n)$, and any $\delta \in (2^{-\alpha t},1/\beta t)$ there is
    a function $f:\zo^{r=\gamma\log 1/\delta} \to \zo^n$ such that
    \begin{enumerate}
        \item $f$ is computable in time $\poly(t)$
        \item for any fixed $z$, $\Pr_{x\gets \cU_r}[\aK^t_z(f(x))\mid x)\le n-\log(1/\delta)-r]\le 3\delta$.
    \end{enumerate}

    e.g.~if we take $t=n^c$ and $\delta=2^{-\sqrt{n}}$, then the resulting $f$ is $(t,n-O(\sqrt{n})$-incompressible with error $\delta=2^{-\sqrt{n}}$. Else if we take $\delta=2^{-\poly\log(n)}=\negl(n)$, then the resulting $f$ is $(t,n-O(\polylog))$ incompressible with $\negl(n)$ error.
\end{corollary}
Combining the above with Corollary~\ref{cor:ECC} and Theorem~\ref{thm:kt-ecc} we get the following corollary that says effectively that for any arbitrarily small constant $\epsilon$ there is a constant $\delta$ such that one can get a polynomial time expanding function which cannot be $\delta$-approximately compressed to size $(1-\epsilon)n$ (relative to a polynomial decompression time).
\begin{corollary}
    \label{cor:approx-incompressible}
    Assume $\mathsf{E}$ is hard for exponential size nondeterministic circuits.

    Then there is a polynomial $q$ such that for any polynomial decompression time $t$, every $\rho \in(0,1/2)$, and any $\mu \in (2^{-q(t+m)},1/q(t+m))$, there is an explicit function
\[
    F:\zo^r\to\zo^m,\qquad r=O(\log(1/\mu)),
\]
computable in time $\poly(t+m)$ such that for any fixed advice $z$
\[
    \Pr_{x\gets\cU_r}[\exists \hat{y}\in\zo^m: \Delta(\hat{y},F(x))<\rho m \wedge \aK^t_z(\hat{y}\mid x)\le s_\rho(m,\mu)]\le 3\mu,
\]
where
\[
s_\rho(m,\mu) = m-H_2(\rho)m-r-\log(1/\mu)-O(1).
\]
In other words, $F$ is $\rho$-approximate $(t,s_\rho(m,\mu)$-incompressible with error $3\mu$.

In particular if we take $\rho=o(1)$, then $H_2(\rho)m=O(\rho m \log(1/\rho)$.
\end{corollary}

In general, if a $G$ is secure against $\gamma$-limited nondeterministic distinguishers for $\gamma<n-\log(1/\delta)$, then one can lower bound the conditional $T$-time-bounded Kolmogorov complexity of $G$'s output (conditioned on the input) by $\gamma$.

We observe that limited nondeterministic hardness is in fact \emph{necessary} for high-on-average $\Kt$ complexity (and hence for high incompressibility). Nondeterministic computation is traditionally studied in the context of decision problems and these definitions are not always amenable to computing many valued functions. In the next proposition, we consider a multivalued nondeterministic computation to be one that outputs a \emph{set} of outputs. An \emph{multivalued $\gamma$-limited nondeterministic computation in time $t$} $N$ corresponds is defined by a deterministic time $t$ computation $M$ that takes an input and a witness of length $\gamma$: $N(x) := \{y: \exists w\in\zo^\gamma, M(x,w)=y\}$. We say such a computation fails to compute function if the output of the function is not contained in its output set. (We note that this is not the only way to define nondeterministic computation for multibit output, merely a simple and strong one that suffices in this context.)
\begin{prop}
    Let $f:\zo^k\to\zo^n$ such that $\Pr[\Kt(f(\cU_k))\ge \ell]>1-\varepsilon$.
    Then, for any multivalued $(\ell-k-\omega(1))$-limited nondeterministic time $t$ computation, $N$:
    \[\Pr_{x\sim\cU_k}[f(x) \notin N(x)]>1- \varepsilon\]
\end{prop}
Viewed this way, multivalued nondeterminism shows we can interpret this proposition as saying that any function with high $\Kt$ complexity outputs must remain hard for time $t$ even when such computation is given arbitrary bounded leakage on the output.
\begin{proof}[Sketch]
    For the sake of contradiction, suppose $N$ is such a limited nondeterministic machine such that
    \[ \Pr_x[\exists w, |w|\le \ell-k \wedge N(x,w)=f(x)]\ge \varepsilon.\]
    Then for any $x$ such that the above event holds, $y=f(x)$ can be described via $\langle N \rangle$, $x$ and $w$. Thus, $\Kt(y) \le O(1)+|x|+|w| = O(1) + k + \ell-k-\omega(1) \le \ell$.
\end{proof}

In summary, one gets a function whose outputs have large computational depth (even when given the input) from fairly standard derandomization circuit lower bounds against nondeterministic circuits. And moreover, a weaker form of lower bounds against such circuits are in fact necessary. We leave constructing PRGs against limited nondeterministic circuits from minimal circuit lower bounds to follow up work. 

Finally we remark that such absolute incompressibility is not necessary in the context of PoS, only computational soundess is required. More precisely, in this setting, we only require that it is hard to find a compression in arbitrary polynomial time (not that no such compression exists).

%% file: extraction.tex
\section{Extraction Lemma}

In this section, we present our core technical lemma, the extraction lemma, which will be a central ingredient in our constructions of Proof of Space.

We first go through some useful lemmas that will be needed for the final extraction lemma.

\subsection{Useful Lemmas}

To prove the extraction lemma, we will need the following two lemmas.

\begin{lemma}
    \label{lem:density}
    Let $A$ be a randomized algorithm and let $\mathbb{X}$ denote a random variable. If $\Pr_{\mathbb{X}}[A(\mathbb{X}) = 1] > \varepsilon$, then
    $$\Pr_{\mathbb{X}}\left[\mathbb{X} \in \good\right] > \frac{\varepsilon}{2},$$
    where $\good = \{x : \Pr[A(x) = 1] > \varepsilon/2\}$ is the set of inputs on which $A$ accepts with probability greater than $\varepsilon/2$.
\end{lemma}

\begin{proof}
    We have
    \begin{align*}
    \varepsilon < \Pr_{\mathbb{X}}[A(\mathbb{X}) = 1] &= \Pr_{\mathbb{X}}[A(\mathbb{X}) = 1 \mid \mathbb{X} \in \good] \cdot \Pr_{\mathbb{X}}[\mathbb{X} \in \good] \\
    &\quad + \Pr_{\mathbb{X}}[A(\mathbb{X}) = 1 \mid \mathbb{X} \notin \good] \cdot \Pr_{\mathbb{X}}[\mathbb{X} \notin \good].
    \end{align*}
    The first conditional probability is at most $1$, and the second is at most $\varepsilon/2$ by definition of $\good$. Thus,
    $$\varepsilon < \Pr_{\mathbb{X}}[\mathbb{X} \in \good] + \frac{\varepsilon}{2}.$$
    Rearranging gives $\Pr_{\mathbb{X}}[\mathbb{X} \in \good] > \varepsilon/2$.
\end{proof}

\begin{lemma}[Index Covering]
    \label{lem:covering}
    Let $\beta=\beta(n)\in(0,1)$ and $\varepsilon\in(0,1)$.

    Let $\cG \subseteq [n]^k$ be a collection of ``good'' $k$-tuples such that $|\cG| \geq \varepsilon \cdot n^k$.
    Define
    \[
    	\Lambda_{\beta,\varepsilon,k} := k\log\frac{1}{1-\beta}-\log\frac{1}{\varepsilon}
    \]
    
    If $\Lambda_{\beta,\varepsilon,k}>0$, then for any $\lambda\ge 1$, $\alpha$ independent samples from the uniform distribution over $\cG$ cover at least $(1-\beta)n$ coordinates except with probability $2^{-\lambda}$, provided
    \[
    	\alpha\ge \frac{H_2(\beta)n+\lambda+2}{\Lambda_{\beta,\varepsilon,k}}
    \]
    In particular, let $\varepsilon \ge n^{-O(1)}$ and $k=\Theta(\log^2(n))$, then taking $\beta=1/\sqrt{\log n}$ gives $\alpha=O(n\log\log n /\log^2 n)$.
\end{lemma}

\begin{proof}
    Fix an arbitrary set $B \subseteq [n]$ of size $|B| = \lceil\beta n\rceil$. Let $\cS_B \subseteq [n]^k$ be the set of $k$-tuples that avoid $B$ entirely, i.e., tuples with all coordinates in $[n] \setminus B$. Then $|\cS_B| \le ((1-\beta)n)^k$.

    Let $S$ be a uniformly random $k$-tuple from $\cG$. The probability that $S$ avoids $B$ is
    \[
    \Pr[S \in \cS_B] = \frac{|\cS_B \cap \cG|}{|\cG|} \leq \frac{|\cS_B|}{|\cG|} \leq \frac{((1-\beta)n)^k}{\varepsilon \cdot n^k} = \frac{(1-\beta)^k}{\varepsilon} = 2^{-\Lambda_{\beta,\varepsilon,k}}.
    \]

    Therefore if we sample $\alpha$ independent $k$-tuples from $\cG$, the probability that all of them avoid $B$ is at most $\left(\frac{(1-\beta)^k}{\varepsilon}\right)^\alpha$. Notice that $$\frac{(1-\beta)^k}{\varepsilon}= 2^{\log(1/\varepsilon)-k\log(1/(1-\beta))}=2^{-\Omega(k)}.$$

    Now, the entropy bound gives (where $H_2$ is the binary entropy function):
    \begin{align*}
        \binom{n}{\lceil \beta n\rceil} \le 2^{H_2(\beta)n+2}
    \end{align*}
    
    Thus, by a union bound over all $\binom{n}{\lceil\beta n\rceil}$ choices of $B$, the probability that there exists some set of $\lceil \beta n\rceil$ indices not covered by any sample is at most
    \[
        2^{H_2(\beta)n+2}\cdot \left(\frac{(1-\beta)^k}{\varepsilon}\right)^\alpha \le 2^{H_2(\beta)n+2-\alpha\Lambda_{\beta,\varepsilon,k}}
    \]
	By the assumed lower bound on $\alpha$, this is at most $2^{-\lambda}$.

\end{proof}

\subsection{Extraction Lemma}

Finally, we present the extraction lemma. This roughly says that if there is a two-phase (polytime initialization that generates a small state to pass to a $t$-time second phase) algorithm that succeeds in violating the Proof of Space guarantees when the verifier has access to the \emph{correct} commitment to the incompressible string $\vy$, then $\vy$ is approximately incompressible (relative to time decompression time $t$). Note that this decompression time is crucially independent of the time to compute $f$, $T$. 

We frame our argument as using advice, however under the derandomization assumption if the adversarial prover $\cA$ is uniform then this advice is dominated by the length of $\params$. (We remark that at the expense of longer advice the derandomization assumption is unnecessary.)

Note that this critically does not say anything about how the verifier got the correct commitment. In Section~\ref{sec:pos-constructions} we show how to use certain succinct proof systems to do just this.

\begin{lemma}[Extraction Lemma]
    \label{lem:extraction}
    Assume $\mathsf{E}=\mathsf{DTIME}(2^{O(n)})$  is hard for exponential size nondeterministic circuits.
    
    Let $\MERKLE=(\Commit, \Open, \Verify)$ be a 2-PLOSC scheme that is $\varepsilon_1$-secure.
    Fix approximation parameter $\beta=\beta(N)\in(0,1/2]$, a query length $k<N$, and $\varepsilon_2\ge N^{-O(1)}$. Let $\gamma:=\varepsilon_2-k\varepsilon_1$ and define
    \[
    	\Lambda:= k\log\frac{1}{1-\beta}-\log\frac{4}{\gamma}.
    \]
    Assume $\gamma\ge N^{-O(1)}$ and $\Lambda>0$ and set
    \[
    	\alpha \ge \lceil \frac{NH_2(\beta)+3}{\Lambda}\rceil.
    \]

    Let $\ell_h$ be an upper bound on the number of bits used to sample $\vh$ in the 2-PLOSC.

    Fix $\cA=(\cA_0, \cA_1)$ where $\cA_0$ is a non-uniform PPT initialization algorithm and $\cA_1$ is a (randomized) non-uniform algorithm running in time $t$. Then, there is a fixed non-uniform advice string $z$ (that is comprised of $\params$, 
    $\cA_1$, 
    and a HSG description) such that:

    If $\vy\in\Sigma^N$ and $x$ such that $f(x)=\vy$,
    \[ 
    \Pr\left[\forall j \in [k],\Verify(\params,c,i_j,y'_j,\vp_j)=1;
    \begin{array}{c}
        \vh\gets \cH_{\text{2-PLOSC}},\\
        (c,\st)\gets \Commit(\params,\vh,\vy),\\
        \sigma\gets \cA_0(\params,\vh,\vy,x,c):\sigma\in\zo^m,\\
        \vi = (i_1,\ldots,i_k)\gets [N]^k,\\ (y_1',\ldots,y_k',\vp_1,\ldots,\vp_k)\gets \cA_1(\params,\vh,\sigma,\vi)
    \end{array}\right]\ge \varepsilon_2,
    \]
    
    Then
    there is a string $\hat{\vy}$ satisfying $\Delta(\vy,\hat{\vy})\le \beta N$ such that 

    \[
        \aK^{t'}_z(\hat{\vy}\mid x) \le m + O\left(\alpha\log\frac{1}{\gamma}+\log (N+t+\ell_h+m)\right)
    \]
    where $t'=\poly(t\alpha+N+\ell_{\vh}+m)$.
    
    In particular, for $\varepsilon_1=\negl(N)$ and $k=\Theta(\log^2 N)$, if we take $\beta=1/\sqrt{\log N}$, then we have $\alpha=O(N\log\log N/\log^2 N)$. Thus,
    \[
    	\Delta(\vy,\hat{\vy})\le O(N/\sqrt{\log N}) \qquad \mbox{and} \qquad \aK^{t'}_z(\hat{\vy}\mid x) \le m + O\left(\frac{N\log\log N}{\log N}+\log(t+\ell_h+m)\right)
    \]

\end{lemma}
Before proving we make some remarks about this lemma, whose proof is more general than the lemma suggests.
\begin{remark}
    The derandomization assumption in this extraction lemma is almost completely superfluous, but we believe using it makes the resulting statement/proof cleaner (and we are already relying on it for incompressibility). This assumption is used solely to instantiate the HSG explicitly. In particular, because the advice string is comprised solely of $\params$ and descriptions of algorithms $\cA$ and the HSG, this actually gives a bound on $\mathrm{K}^t(\hat{y}\mid x)$, with a (additive) $|\params|+O(1)$ loss over what is in the lemma. 
    
    On the other hand, at the expense of longer advice $z$, one can instead instantiate the HSG unconditionally using lemma~\ref{lem:mc-hsg}. Because contradicting incompressibility only requires compression of a non-negligible fraction, one need not compress the witness set of any circuit $C_{x,\vy}$ but only a significant fraction of them (see proof below for description of these circuits). Hence, one will get that for typical $x$ on which $\cA$ is successful, even better parameters than above:
    \[\aK^{t'}_z(\hat{\vy}|x)\le m+O\left(\alpha\log(1/\gamma)+\log N\right).\]
\end{remark}
\begin{proof}
    Call an execution of $\cA$ with commitment (root) $c$ when given a tuple $\vi=(i_1,\ldots,i_k)$ \emph{accepting} if all $k$ openings verify against $c$ and \emph{correct} if $y'_j=y_{i_j}$ for every $j\in[k]$. 
    
    Intuitively, the binding of the 2-PLOSC tells us that correct and accepting executions must occur nearly as frequently as accepting executions. In particular, if an execution accepts but is not correct then for some coordinate $j$ it gives an accepting opening to $c$ whose value is not $y_j$, which violates soundness. Hence by a union bound over the $k$ coordinates and the $\varepsilon_1$-security of the 2-PLOSC, such an execution happens with probability at most $k\varepsilon_1$. Hence, the probability of a correct \emph{and} accepting execution is at least $\gamma=\varepsilon_2-k\varepsilon_1$.

    Now, for each $\vh$, let $(c_{\vh},\st_{\vh})= \Commit(\params,\vh,\vy)$ and define
    \[
        p_{\vh}:= \Pr_{\sigma\gets \cA_0(\params,\vh,\vy,x,c_{\vh}),\vi,\cA_1} \left[\cA_1(\params,\vh,\sigma,\vi) \mbox{ is correct and accepting with root $c_{\vh}$}\right].
    \]
    By Lemma~\ref{lem:density}, because $\bbE_{\vh}[p_{\vh}]\ge \gamma$,\footnote{In other words, $\Pr_{\vh}[\cA^*(\vh)=1]\ge \gamma$ where $\cA^*$ is the randomized procedure on the RHS above such that $\Pr[\cA^*(\vh)=1]=p_{\vh}$.}
    it follows that if
    \[ 
        \good_{x,\vy}:=\{\vh\mid p_{\vh}\ge \gamma/2\},
    \]
    then,
    \[
    \Pr_{\vh}[\vh\in\good]\ge \gamma/2.
    \]
    For any $\vh\in\good$, by an averaging argument over $\cA_0$ there must be some $\sigma_{\vh}$ of length $m$ such that
    \[
        \Pr_{\vi}[\cA_1(\params,\vh,\sigma_{\vh},\vi) \mbox{ is correct and accepting with root }c_{\vh}]\ge \gamma/2.
    \]
    Fix some such $\sigma_{\vh}$ for each $\vh$. Now given this we can define 
    \[
        \cG_{\vh}:= \left\{\vi \mid \Pr[\cA_1(\params,\vh,\sigma_{\vh},\vi)\mbox{ is correct and accepting with root }c_{\vh}]\ge \gamma/4 \right\}.
    \]
    Again invoking Lemma~\ref{lem:density}, we have $$\Pr[\vi\in \cG_{\vh}]\ge \gamma/4.$$
    
    Apply Lemma~\ref{lem:covering} with $\varepsilon=\gamma/4$ and $\lambda=1$. By the bound on $\alpha$, if independent $\vi_1,\ldots,\vi_\alpha$ are sampled from $\cG_{\vh}$ then they will cover $(1-\beta)$-fraction of $[N]$ with probability at least $1/2$. Hence if the $\vi_1,\ldots,\vi_\alpha$ are sampled uniformly from the entire space, $[N]^k$, then the probability all hit $\cG_{\vh}$ and cover $(1-\beta)$-fraction of $[N]$ is at least
    \[
    	\frac{1}{2}(\gamma/4)^\alpha
    \]
    Now consider the set of strings comprised of the hash, tuples, and coins for $\cA_1$ for each invocation,
    \[
        \omega=(\vh,\vi_1,\ldots,\vi_\alpha,r_1,\ldots,r_\alpha).
    \]
    Let $\great_{x,\vy}$ denote the set of pairs $(\sigma,\omega)$ such that $\vi_1,\ldots,\vi_\alpha$ in $\cG_{\vh}$ and cover a $(1-\beta)$ fraction, and all invocations of $\cA_1$ on the respective coins and state $\sigma$ are correct and accepting. Then by the discussion above,
    \[
        \Pr_{(\sigma,\omega)}[(\sigma,\omega)\in\great]\ge 2^{-m}\cdot \underbrace{\gamma/4 (\gamma/4)^{2\alpha}}_{p^*:=\Pr_\omega[\exists \sigma,(\omega,\sigma)\in\great]}.
    \]
    Now consider the circuit $C_{x,\vy}$ that has $x$ and $\vy$ hardwired and recognizes $\great_{x,\vy}$. It does the following on input $(\sigma,\omega)$:
    \begin{enumerate}
        \item Parse $\omega$ into $\vh, \vi_1,\ldots, \vi_\alpha,r_1,\ldots,r_\alpha$.
        \item Compute $(c,\st)=\Commit(\params,\vh,\vy)$.
        \item Simulate $\alpha$ runs of $\cA_1$ using $x$, state $\sigma$, indices $\vi_j$ and randomness $r_j$ (for $j=1,\ldots,\alpha$), to get $\hat{\vy}_{\vi_1},\ldots,\hat{\vy}_{\vi_\alpha},\hat{\pi}_1,\ldots,\hat{\pi}_\alpha$.
        \item Check that all proof, output pairs $(\hat{\vy}_{\vi_j},\hat{\pi}_j)$ verify.
        \item Check that all symbols in $\hat{\vy}_{\vi_1},\ldots,\hat{\vy}_{\vi_\alpha}$ match $\vy$.
        \item Check that $\vi_1,\ldots,\vi_\alpha$ cover $(1-\beta)$ fraction of $[N]$.
        \item Accept iff all checks pass.
    \end{enumerate}
    Note that $C_{x,\vy}$ is $\poly(t\alpha+N+\ell_h+m)$-size so we can hit a witness with a good Hitting Set Generator (HSG). Critically, this is independent of $T$, the time to compute $f$. Applying Lemma~\ref{lem:hsg} to $C_{x,\vy}$ with error $2^{-m}p^*/2$ yields a $\poly(t\alpha+N+\ell_h+m)$-time function $G$ and seed $s^*$ of length
    \[
        m+O\left(\log(1/p^*)+\log\left(\alpha(k\log N +t) + N + \ell_{\vh}+m\right)\right) = m+O\left(\alpha \log(1/\gamma)+\log(N+t+\ell_h+m)\right)
    \]
    such that if $G(s^*)=(\sigma^*,\omega^*)$ then $C_{x,\vy}(\sigma^*,\omega^*)=1$. Hence we can use $\sigma^*,\omega^*$ and $\cA_1$ to recover all but $\beta$-fraction of $\vy$.
    
    Thus, we can define a decompressor $D_{z}$ with advice $z$ that does just that. The advice $z$ is comprised of $\params$, description of $\cA_1$, the HSG $G$'s description. On input $(x,s^*)$, $D_z$ does the following:
    \begin{enumerate}
        \item Expand $G(s^*)=(\sigma^*,\omega^*)$ and parse $\omega^*=(\vh^*,\vi_1^*,\ldots, \vi_\alpha^*,r_1^*,\ldots,r_\alpha^*)$.
        \item For each $j\in [\alpha]$, compute $\cA_1(\params,\vh^*,\sigma^*,\vi_j^*;r_j^*)=\hat{\vy}_{\vi^*_j},\pi_j$.
        \item Set $\hat{\vy}$ to be consistent with $\hat{\vy}_{\vi_1^*},\ldots,\hat{\vy}_{\vi_\alpha^*}$ and 0 otherwise.
        \item Return $\hat{\vy}$.
    \end{enumerate}
    Because $(\sigma^*,\omega^*)\in\great_{x,\vy}$ (and hence $\vi_j$'s cover all but $\beta$-fraction and all $\hat{\vy}_{\vi_{j}}$'s returned by $\cA_1$ match $\vy$), we know that the $\hat{\vy}$ returned is such that $\Delta(\vy,\hat{\vy})\le \beta N$.

    Moreover by inspection,
    \[
        \aK^{t'}_z(\hat{\vy}|x)\le m+O\left(\alpha\log(1/\gamma) + \log(N+t+\ell_{\vh}+m)\right)
    \]
    for $t'=\poly(t\alpha+N+\ell_h+m)$.
\end{proof}

%% file: pos-2.tex
\section{Proof of Space Constructions}\label{sec:pos-constructions}

This section instantiates the framework in two ways, building on prior sections.

Throughout, let $n=\lout$ be the output length of $f$.  In this section we use
one concrete extraction regime:
\[
    k=\Theta(\log^2 n),
    \qquad
    \beta(n)=\frac1{\sqrt{\log n}} .
\]
The hidden constant in $k$ is chosen large enough for every fixed
inverse-polynomial residual success density.  We assume that $f$ is
$2\beta$-approximately $(t_{\mathsf K},s)$-incompressible with error
$\varepsilon_{\mathsf{inc}}$.  Thus a description that outputs any string
within distance $\beta n$ of $f(x)$ for an $\varepsilon_{\mathsf{inc}}$-dense
set of inputs already contradicts the assumed incompressibility of $f$.

We use the following root-checking computation in both instantiations.  The
machine $M_f$ takes $(\params,x,\mathbf h,c)$ as configuration data, computes
$y=f(x)$, computes the Merkle root
\[
    (\hat c,\bot)=\MERKLE.\Commit(\params,h_1,\ldots,h_d,y),
\]
and accepts if and only if $\hat c=c$.  Its running time is denoted by
$T_f$.

\subsection{Construction from Non-Interactive Arguments}

Our first construction uses a semi-universal non-interactive argument for the
root-checking computation above.  The timing of the commitment is important:
the prover does not send the root in the initialization phase.  The verifier
first sends the random positions to be opened, and the prover then responds with
the root, the openings, and the argument that the root is the honest root of
$f(x)$.

\begin{construction}[PoS from semi-universal arguments]\label{con:fromarg}
Let $\MERKLE$ be an $\varepsilon_{\mathsf{mt}}$-secure $2$-PLOSC scheme and let
$\ARG=(\Prove,\Verify)$ be a semi-universal non-interactive argument with
soundness error $\varepsilon_{\mathsf{arg}}$ for computations of time $T_f$.
The protocol is as follows.

\begin{itemize}
    \item \textbf{Initialization phase.}
    \begin{mdframed}
    Verifier $V_1(\params)$:
    \begin{enumerate}
        \item Sample $x\from\zo^{\lin}$ and hash functions
        $\mathbf h=(h_1,\ldots,h_d)$ for $\MERKLE$.
        \item Send $(x,\mathbf h)$ to $P_1$.
        \item Output verifier state $\Phi=(x,\mathbf h)$.
    \end{enumerate}
    \end{mdframed}

    \begin{mdframed}
    Prover $P_1(\params)$:
    \begin{enumerate}
        \item Receive $(x,\mathbf h)$.
        \item Compute $y=f(x)$ and
        $(c,\st)=\MERKLE.\Commit(\params,h_1,\ldots,h_d,y)$.
        \item Compute
        \[
        \pi\from
        \ARG.\Prove\!\left(\crs,
        (M_f,\cf=(\params,x,\mathbf h,c),\cf'=\accept,T_f)\right).
        \]
        \item Store $S=(c,\st,\pi)$.
    \end{enumerate}
    \end{mdframed}

    \item \textbf{Execution phase.}
    \begin{mdframed}
    Verifier $V_2(\params,\Phi)$:
    \begin{enumerate}
        \item Parse $\Phi=(x,\mathbf h)$.  Sample and send
        $\nu=(i_1,\ldots,i_k)\from[n]^k$.
        \item Receive $c$, $\pi$, and openings
        $\{(y_{i_j},\mathbf p_j)\}_{j\in[k]}$.
        \item Check
        \[
        \ARG.\Verify\!\left(\crs,
        (M_f,\cf=(\params,x,\mathbf h,c),\cf'=\accept,T_f),\pi\right)=1.
        \]
        \item For every $j\in[k]$, check
        $\MERKLE.\Verify(\params,c,i_j,y_{i_j},\mathbf p_j)=1$.
        \item Accept if and only if all checks accept and the response arrives
        within the execution time bound.
    \end{enumerate}
    \end{mdframed}

    \begin{mdframed}
    Prover $P_2(\params,S)$:
    \begin{enumerate}
        \item Parse $S=(c,\st,\pi)$ and receive
        $\nu=(i_1,\ldots,i_k)$.
        \item For each $j\in[k]$, compute
        $\mathbf p_j\from\MERKLE.\Open(\params,\st,c,i_j)$.
        \item Send $c$, $\pi$, and
        $\{(y_{i_j},\mathbf p_j)\}_{j\in[k]}$.
    \end{enumerate}
    \end{mdframed}
\end{itemize}
\end{construction}

\begin{theorem}[Soundness from semi-universal arguments]\label{thm:snarg-pos}
Assume the derandomization hypothesis required by Lemma~\ref{lem:extraction}.

Let $k=\Theta(\log^2 n)$ and $\beta=1/\sqrt{\log n}$ be the parameters fixed
above.

Let $f:\zo^{\lin}\to\zo^n$ be computable in time $T_f$ and be
$(\beta+\eta)$-approximately
$(t_{\mathsf K},s)$-incompressible with error
$\varepsilon_{\mathsf{inc}}$.  Let $\MERKLE$ be
$\varepsilon_{\mathsf{mt}}$-secure and let $\ARG$ have soundness
$\varepsilon_{\mathsf{arg}}$ for the root-checking instances above.

For every execution time bound $\tau$ and every parameter
$\rho>0$, let $\Delta_\rho$ denote the additive storage loss guaranteed by
Lemma~\ref{lem:extraction}, at radius $\beta$ and query count $k$, when applied
with success parameter
$\theta_\rho:=\varepsilon_{\mathsf{inc}}+\rho/2$, Merkle error
$\varepsilon_{\mathsf{mt}}$, second-phase time $\tau$, and state-size bound
$m\le s$.  Suppose this application of Lemma~\ref{lem:extraction} is in its
stated parameter range and has decompressor running time at most
$t_{\mathsf K}$.  

Then
Construction~\ref{con:fromarg}, instantiated with $k$ execution queries, is
complete, has efficient verification and honest proving time polynomial in
$T_f+n$, and is
\[
    \left(s-\Delta_\rho,\ \tau,\ 
    \varepsilon_{\mathsf{arg}}+2\varepsilon_{\mathsf{inc}}+\rho\right)
    \text{-sound}.
\]
\end{theorem}

\begin{proof}
Completeness follows from correctness of $\MERKLE$ and completeness of
$\ARG$.  The efficiency bounds are immediate from the construction: $P_1$
computes $f$, one Merkle commitment, and one argument for a time-$T_f$
computation; $P_2$ opens $k$ Merkle paths and sends one already-computed
argument; and $V_2$ verifies one succinct argument and $k$ Merkle paths.

It remains to prove soundness.  Suppose, toward contradiction, that a prover
$\cA=(\cA_1,\cA_2)$ with initialization state size at most
$s-\Delta_\rho$ and execution time at most $\tau$ makes
$V_2$ accept with probability
\[
    \mu>
    \varepsilon_{\mathsf{arg}}+2\varepsilon_{\mathsf{inc}}+\rho.
\]
Set $\theta_\rho=\varepsilon_{\mathsf{inc}}+\rho/2$.
The probability is over $x$, $\mathbf h$, the verifier's sampled tuple
$\nu\in[n]^k$, and the prover's randomness.

Let $\mathsf{FalseArg}$ be the event that $V_2$ accepts but the statement
accepted by $\ARG$ is false, namely
\[
    \MERKLE.\Commit(\params,h_1,\ldots,h_d,f(x))
\]
does not have root equal to the $c$ sent by the prover in the execution phase.
Although $c$ is sent only after the verifier samples $\nu$, an adversary for $\ARG$ can internally:
run the whole PoS initialization, sample $\nu$, run $\cA_2$ to obtain
$(c,\pi)$, and output the false root-checking instance together with $\pi$.
Therefore
\[
    \Pr[\mathsf{FalseArg}]\le \varepsilon_{\mathsf{arg}}.
\]
Let $\mathsf{TrueAcc}$ be the event that $V_2$ accepts and
$\mathsf{FalseArg}$ does not occur.  Then
\[
    \Pr[\mathsf{TrueAcc}]
    >
    \mu-\varepsilon_{\mathsf{arg}}
    >
    2\theta_\rho.
\]

For a fixed input $x$, let
$q_x=\Pr[\mathsf{TrueAcc}\mid x]$.  Define
$X_{\mathsf{good}}=\{x:q_x>\theta_\rho\}$.  By Lemma~\ref{lem:density},
applied with threshold $2\theta_\rho$, the set $X_{\mathsf{good}}$ has density
greater than $\theta_\rho$, and hence greater than
$\varepsilon_{\mathsf{inc}}$.  Fix any $x\in X_{\mathsf{good}}$ and write
$y=f(x)$.

We now instantiate Lemma~\ref{lem:extraction}.  The extraction initialization
algorithm runs $\cA_1$ on $(x,\mathbf h)$ and outputs its stored state.  The
query algorithm, on a tuple $\nu$, runs $\cA_2$ on that state after sending
$\nu$, parses the response as
$(c,\pi,\{(a_j,\mathbf p_j)\}_{j\in[k]})$, discards $c$ and $\pi$, and outputs
the claimed symbols and Merkle paths.  Whenever $\mathsf{TrueAcc}$ occurs, the
argument statement is true, so the root $c$ sent by the prover is exactly the
honest root of $y=f(x)$ under $\mathbf h$, and all $k$ Merkle openings verify
under this honest root.  The bad-opening probability required by
Lemma~\ref{lem:extraction} is bounded by $k\varepsilon_{\mathsf{mt}}$, by the
$2$-PLOSC security of $\MERKLE$ and a union bound over the $k$ opened
positions.  Thus the lemma applies with success parameter $\theta_\rho$ and
residual density
\[
    \gamma_\rho
    :=
    \theta_\rho-k\varepsilon_{\mathsf{mt}}.
\]

Consequently, there exists an advice $z$ such that for every $x\in X_{\mathsf{good}}$, there is a string
$\hat y$ such that
$\Delta(y,\hat y)\le\beta n$ and
\[
    \aK^{t_{\mathsf K}}_z(\hat y\mid x)
    \le (s-\Delta_\rho)+\Delta_\rho
    \le s.
\]
Since $X_{\mathsf{good}}$ has density greater than
$\varepsilon_{\mathsf{inc}}$ and
$\beta<\beta+\eta$, this contradicts the
assumed approximate incompressibility of $f$.  Hence no such prover exists.
\end{proof}

\begin{corollary}[Tight PoS from SNARGs]\label{cor:snarg-near-tight}
Assume the hypotheses of Theorem~\ref{thm:snarg-pos}.  Instantiate $f$ using
Corollary~\ref{cor:approx-incompressible} so that it is
$2/\sqrt{\log n}$-approximately incompressible against descriptions of length
$n-O(n\log\log n/\sqrt{\log n})$, with negligible error.  Assume also that the
argument and Merkle soundness terms are negligible and that extraction runs
within the incompressibility time bound.  Then Construction~\ref{con:fromarg}
is sound against every time-$\tau$ prover storing at most
\[
    n-C\frac{n\log\log n}{\sqrt{\log n}}
    =
    (1-o(1))n
\]
bits, for a sufficiently large constant $C$, with only negligible accepting
probability. 
\end{corollary}

\subsection{Construction from Interactive Arguments}

The second instantiation replaces the SNARG by a Kilian-style PCP argument.
This reduces the cryptographic assumption on the argument system to
collision-resistant hashing, but the honest prover must store the PCP string
for the root-checking computation which has size quasipolynomial in the time to compute $f$.

\begin{construction}[PoS from Kilian-style arguments]\label{con:fromKilian}
Let $\MERKLE_1$ be the $2$-PLOSC scheme used to commit to $y=f(x)$, and let
$\MERKLE_2$ be the position-binding vector-commitment used in
Kilian's protocol for committing to a PCP string.  Let $\PCP=(P,V)$ be a PCP for the relation
$R_{\mathsf{root}}$ whose instances are
\[
    u=(M_f,\cf=(\params,x,\mathbf h,c_1),\cf'=\accept,T_f),
\]
and whose witnesses are accepting computations of $M_f$.

\begin{itemize}
    \item \textbf{Initialization phase.}
    \begin{mdframed}
    Verifier $V_1(\params)$:
    \begin{enumerate}
        \item Sample $x\from\zo^{\lin}$ and hash functions
        $\mathbf h=(h_1,\ldots,h_d)$ for $\MERKLE_1$.
        \item Send $(x,\mathbf h)$ to $P_1$.
        \item Output verifier state $\Phi=(x,\mathbf h)$.
    \end{enumerate}
    \end{mdframed}

    \begin{mdframed}
    Prover $P_1(\params)$:
    \begin{enumerate}
        \item Receive $(x,\mathbf h)$ and compute $y=f(x)$.
        \item Compute
        $(c_1,\st_1)=\MERKLE_1.\Commit(\params,h_1,\ldots,h_d,y)$.
        \item Form
        $u=(M_f,\cf=(\params,x,\mathbf h,c_1),\cf'=\accept,T_f)$ and
        compute the PCP string
        $\Pi\from\PCP.P(u,w)$, where $w$ is the accepting computation of
        $M_f$.
        \item Compute
        $(c_2,\st_2)=\MERKLE_2.\Commit(\params,\Pi)$.
        \item Store $S=(c_1,c_2,\st_1,\st_2,\Pi)$.
    \end{enumerate}
    \end{mdframed}

    \item \textbf{Execution phase.}
    \begin{mdframed}
    Verifier $V_2(\params,\Phi)$:
    \begin{enumerate}
        \item Parse $\Phi=(x,\mathbf h)$.  Sample and send
        $\nu=(i_1,\ldots,i_k)\from[n]^k$.
        \item Receive commitments $c_1,c_2$ and openings
        $\{(y_{i_j},\mathbf p_j)\}_{j\in[k]}$.
        \item For every $j\in[k]$, check
        $\MERKLE_1.\Verify(\params,c_1,i_j,y_{i_j},\mathbf p_j)=1$.
        \item Form
        $u=(M_f,\cf=(\params,x,\mathbf h,c_1),\cf'=\accept,T_f)$.
        Sample PCP randomness $\rho$ and send it to $P_2$.
        \item Receive $(\cQ,\ans,\pf)$.
        \item Check that $\MERKLE_2.\Verify(\params,c_2,\cQ,\ans,\pf)=1$
        and that $\PCP.V^{[\cQ,\ans]}(u;\rho)=1$, where the verifier rejects
        if the PCP verifier queries outside $\cQ$.
        \item Accept if and only if all checks accept and the response arrives
        within the execution time bound.
    \end{enumerate}
    \end{mdframed}

    \begin{mdframed}
    Prover $P_2(\params,S)$:
    \begin{enumerate}
        \item Parse $S=(c_1,c_2,\st_1,\st_2,\Pi)$ and receive
        $\nu=(i_1,\ldots,i_k)$.
        \item For every $j\in[k]$, compute
        $\mathbf p_j\from\MERKLE_1.\Open(\params,\st_1,c_1,i_j)$.
        \item Send $c_1,c_2$ and
        $\{(y_{i_j},\mathbf p_j)\}_{j\in[k]}$.
        \item Receive PCP randomness $\rho$.  Compute the PCP query set
        $\cQ$ made by $\PCP.V(u;\rho)$ and compute
        $\pf\from\MERKLE_2.\Open(\params,\st_2,c_2,\cQ)$.
        \item Send $(\cQ,\ans=\Pi[\cQ],\pf)$.
    \end{enumerate}
    \end{mdframed}
\end{itemize}
\end{construction}

\begin{theorem}[Soundness from Kilian-style arguments]\label{thm:kilian-pos}
Assume the derandomization hypothesis required by Lemma~\ref{lem:extraction}.

Let $k=\Theta(\log^2 n)$ and $\beta=1/\sqrt{\log n}$ be the parameters fixed
above.

Let $f:\zo^{\lin}\to\zo^n$ be computable in time $T_f$ and be
$(\beta+\eta)$-approximately
$(t_{\mathsf K},s)$-incompressible with error
$\varepsilon_{\mathsf{inc}}$.  Suppose $\MERKLE_1$ is
$\varepsilon_{\mathsf{mt}}$-secure and the Kilian protocol obtained from
$\PCP$ and $\MERKLE_2$ has soundness error
$\varepsilon_{\mathsf{kil}}$ for relation $R_{\mathsf{root}}$.

For every execution time bound $\tau$ and every parameter
$\rho>0$, let $\Delta_\rho$ denote the additive storage loss guaranteed by
Lemma~\ref{lem:extraction}, at radius $\beta$ and query count $k$, when applied
with success parameter $\theta_\rho=\varepsilon_{\mathsf{inc}}+\rho/2$, Merkle
error $\varepsilon_{\mathsf{mt}}$, second-phase time $\tau$, and state-size
bound $m\le s$.  Suppose this application of Lemma~\ref{lem:extraction} is in
its stated parameter range and has decompressor running time at most
$t_{\mathsf K}$.  

Then Construction~\ref{con:fromKilian}, instantiated with $k$ execution queries, is
complete, has efficient verification and honest proving time polynomial in
$T_f+n+|\Pi|$, and is
\[
    \left(s-\Delta_\rho,\ \tau,\ 
    \varepsilon_{\mathsf{kil}}+2\varepsilon_{\mathsf{inc}}+\rho\right)
    \text{-sound}.
\]
The honest prover's persistent state is $n+|\Pi|+\polylog n$ bits which is dominated by the proof length of $O(T\poly\log T)$.
\end{theorem}

\begin{proof}
Completeness and efficiency follow from correctness of the two commitments and
the PCP.  The only point about message order is that $c_1$ and $c_2$ are sent
after the Merkle index tuple but before the PCP randomness $\rho$; hence the
Kilian commitment to the PCP string is fixed before the PCP challenge, exactly
as required in the standard soundness proof.

The soundness proof is parallel to the proof of
Theorem~\ref{thm:snarg-pos}.  Let a cheating prover with state at most
$s-\Delta_\rho$ and execution time at most $\tau$ make the verifier accept with
probability
\[
    \mu>
    \varepsilon_{\mathsf{kil}}+2\varepsilon_{\mathsf{inc}}+\rho .
\]
Set $\theta_\rho=\varepsilon_{\mathsf{inc}}+\rho/2$.  Let
$\mathsf{FalseKil}$ be the event that the verifier accepts although the
root-checking instance
$u=(M_f,\cf=(\params,x,\mathbf h,c_1),\cf'=\accept,T_f)$ is false.

This event breaks the Kilian argument with the same probability up to its
soundness error.  The reduction samples the PoS initialization and the Merkle
tuple $\nu$, runs the cheating prover until it sends $c_1,c_2$ and the claimed
Merkle openings, sets the Kilian instance to the corresponding $u$, sends
$c_2$ as the PCP commitment, and then forwards the external PCP randomness to
the prover.  Whenever $\mathsf{FalseKil}$ occurs, the reduction convinces the
Kilian verifier of a false instance.  Therefore
\[
    \Pr[\mathsf{FalseKil}]\le\varepsilon_{\mathsf{kil}}.
\]
Let $\mathsf{TrueAcc}$ be acceptance outside $\mathsf{FalseKil}$.  Then
$\Pr[\mathsf{TrueAcc}]>2\theta_\rho$.

As before, define $q_x=\Pr[\mathsf{TrueAcc}\mid x]$ and
$X_{\mathsf{good}}=\{x:q_x>\theta_\rho\}$.  By Lemma~\ref{lem:density}, this
set has density greater than $\theta_\rho$, and hence greater than
$\varepsilon_{\mathsf{inc}}$.  For each fixed $x\in X_{\mathsf{good}}$, set
$y=f(x)$ and invoke Lemma~\ref{lem:extraction} with an initialization
algorithm that runs the cheating initialization and a query algorithm that
sends the tuple $\nu$, receives $c_1,c_2$ and the Merkle openings, discards the
PCP transcript, and outputs the claimed symbols and paths.  On every
$\mathsf{TrueAcc}$ execution, the root-checking instance is true, so $c_1$ is
the honest Merkle root of $f(x)$ under $\mathbf h$, and the displayed openings
verify under the honest root.  The only remaining failure mode for the
extraction lemma is a bad opening, bounded by $k\varepsilon_{\mathsf{mt}}$ by
$2$-PLOSC security of $\MERKLE_1$.  Thus the extraction lemma is invoked with
residual density
\[
    \gamma_\rho
    =
    \theta_\rho-k\varepsilon_{\mathsf{mt}}.
\]

Lemma~\ref{lem:extraction} therefore gives that there exists an advice string $z$, under which for every
$x\in X_{\mathsf{good}}$, there exists a string $\hat y$ with
$\Delta(f(x),\hat y)\le\beta n$ and
$\aK^{t_{\mathsf K}}_z(\hat y\mid x)\le s$.  Since
$X_{\mathsf{good}}$ has density greater than
$\varepsilon_{\mathsf{inc}}$, this contradicts the
$(\beta+\eta)$-approximate incompressibility of $f$.
\end{proof}

\begin{remark}
    \label{rem:pcp-for-time}
    Since $M_f$ runs in time $T_f$, quasilinear-size
    PCPs~\cite{BenSassonSudan08,Dinur07} for $R_{\mathsf{root}}$ give proof
    length $|\Pi|=T_f^{1+o(1)}$ with $\polylog(T_f)$ queries and randomness
    and constant soundness error. Repeating the challenge
    $\lambda=\omega(\log n)$ times against the same commitment $c_2$ (fixed
    before $\rho$ by the message ordering of
    Construction~\ref{con:fromKilian}) yields Kilian soundness
    $\varepsilon_{\mathsf{kil}}=\negl(n)$, with proof length unchanged.
\end{remark}

\begin{corollary}[Kilian space bound with almost-linear PCPs]\label{cor:kilian-space}
Instantiate the PCP in Construction~\ref{con:fromKilian} as in
Remark~\ref{rem:pcp-for-time}, with negligible Kilian soundness.  Under the
same $\beta=1/\sqrt{\log n}$ parameters as
Corollary~\ref{cor:snarg-near-tight}, every time-$\tau$ cheating prover storing
at most
\[
    n-C\,\frac{n\log\log n}{\sqrt{\log n}}
\]
bits succeeds with only negligible probability.  If $N_{\mathsf{hon}}$ is the
honest prover's persistent storage, then
\[
    N_{\mathsf{hon}}=n+T_f^{1+o(1)}+\polylog (n) .
\]
Thus if $T_f=n^{1+o(1)}$ then an adversarial storage lower bound
$N_{\mathsf{hon}}^{1-o(1)}$; more generally, $T_f=n^a$ gives
$N_{\mathsf{hon}}^{1/a-o(1)}$.
\end{corollary}